\documentclass[a4paper,12pt]{article}

\usepackage{amsmath,amsthm,amssymb}
\usepackage[dvips]{graphicx}
\usepackage{multirow}
\usepackage{here}
\usepackage{bm}
\usepackage{float}
\usepackage{latexsym}

\newtheorem{Alemm}{Lemma}

\begin{document}
\setlength{\baselineskip}{18.5pt}
\author{Toshihiro Hirano\footnote{Graduate School of Economics, The University of Tokyo, Hongo 7-3-1, 
Bunkyo-ku, Tokyo 113-0033, 
Japan, Currently at NEC Corporation, 1753, Shimonumabe, Nakahara-ku, Kawasaki, Kanagawa 211-8666, Japan. E-mail: nennekog@hotmail.co.jp}
}
\title{Pseudo best estimator by a separable approximation of spatial covariance structures}
\maketitle

\begin{abstract}

We consider a linear regression model with a spatially correlated error term on a lattice. When estimating  coefficients in the linear regression model, the generalized least squares estimator (GLSE) is used if the covariance structures are known. However, the GLSE for large spatial data sets is computationally expensive, because it involves  inverting the covariance matrix of error terms from each observations. To reduce the computational complexity, we propose a pseudo best estimator (PBE) using spatial covariance structures approximated by separable covariance functions. We derive the asymptotic covariance matrix of the PBE and compare it with those of the least squares estimator (LSE) and the GLSE through some simulations. 
Monte Carlo simulations demonstrate that the PBE using separable covariance functions has superior accuracy to that of the LSE, which does not contain the information of the spatial covariance structure, even if the true process has an isotropic Mat\'ern covariance function. Additionally, 
our proposed PBE is computationally efficient relative to the GLSE for large spatial data sets.

\end{abstract}


\section{Introduction}
\label{intro}

Recently, various statistical methods for spatial data have been investigated. Among them, a linear regression model with a spatially correlated error term has played an important role in a wide variety of scientific fields such as geostatistics, econometrics and forestry.

To estimate the coefficients of the linear regression model with a spatially dependent error term, we often use the generalized least squares estimator (GLSE) (see, e.g., Cressie 1993). However, the GLSE for large spatial data sets is computationally expensive, because it involves inverting the covariance matrix of error terms at $n$ different spatial points, which requires $O(n^3)$ operations. For example, the Walker Lake dataset in Isaaks and Srivastava (1989) and the soil moisture index derived from other GIS layers in Huang et al. (2010) are observed on a spatial lattice of $260 \times 300$ and $100 \times 100$ regular grid points respectively. Recently, much attention has been paid to the analysis of such large spatial datasets (Furrer et al. 2006; Banerjee et al. 2008; Cressie and Johannesson 2008; Lindgren et al. 2011). 

To deal with this problem, we propose a pseudo best estimator (PBE) using separable covariance structures which  approximate the true ones. Specifically, we expressed the covariance function as the product of the covariance functions of the causal autoregressive process. The covariance matrix of the PBE has a separable structure and this  facilitates the computational procedure for large spatial datasets using the property of the Kronecker product. Moreover, unlike the PBE, the GLSE is infeasible because the true covariance matrix is unknown in practice, meaning  that the PBE is also a useful estimator for real spatial data. Genton (2007) considered separable approximation to predict large space-time datasets by solving the nearest Kronecker product for a space-time covariance matrix problem and obtained good predictive performance. 

The main contributions of this paper are to propose the PBE using separable covariance functions for efficient computation and derive its asymptotic covariance matrix by extending the technical method of Grenander and Rosenblatt (1957) and Anderson (1971) in time series to a lattice process. As this analysis provides the condition for the asymptotic efficiency of the least squares estimator (LSE) relative to the GLSE, it is easy in the simulation to compare the PBE with the LSE which is computationally efficient as with the PBE. Our work can be regarded as an extension of Amemiya (1973) and Engle (1974), who considered the asymptotic properties of the GLSE when the covariance structure of the true process was incorrectly identified in time series literature. Koreisha and Fang (2001) investigated the finite accuracy of the GLSE and the PBE for time series. 

In this paper, we compare the finite accuracy and asymptotic variance of the PBE with those of the LSE and the GLSE in Yajima and Matsuda (2008) through some simulations. In these simulations, the effect of the misspecification of the covariance function for the GLSE is also examined. The LSE is efficient in terms of calculating the estimator for large spatial datasets and Yajima and Matsuda (2008) obtained the necessary and sufficient conditions under which the asymptotic covariance matrix of the LSE is identical to that of the GLSE. However, our simulations illustrate that the PBE outperforms the LSE when these conditions are not satisfied and shows good performance as well as the LSE even if these conditions hold. In the simulations, the difference in mean squared error between the PBE approximated by separable covariance functions and the GLSE is small, even in the case of the true process with an isotropic Mat\'ern covariance function. Additionally, our proposed PBE is computationally efficient relative to the GLSE. 


The remainder of this paper is organized as follows. We introduce a linear regression model with a spatially correlated error term and propose the PBE using separable covariance functions in Section 2. In Section 3, we review the necessary and sufficient conditions under which the asymptotic covariance matrix of the LSE is identical to that of the GLSE and derive the asymptotic covariance matrix of the PBE. In Section 4, computer experiments are conducted to examine the finite sample performance of our asymptotic result and compare the finite accuracy and asymptotic variance of the PBE with those of the LSE and GLSE. Our conclusions and future studies are discussed in Section 5. Technical proofs of the lemmas and theorem are given in Appendices A and B. 

\section{Linear regression model and some estimators}

For simplicity, we will consider the sampling region to be a square on a plane. Define $\mathbb{Z}=\{ 0, \pm 1, \pm 2, \ldots \}$. Let $\mathbb{Z}^2$ be the integer lattice points in the two-dimensional Euclidean space. For $\bm{t} = (t_1,t_2)^{'} (\in \mathbb{Z}^2)$, consider the regression model of the form
\[
y_{\bm{t}} = X_{\bm{t}}^{'} \bm{\beta} + \epsilon_{\bm{t}},
\]
where $\{ y_{\bm{t}} \}$ is an observed sequence, $X_{\bm{t}} = (x_{\bm{t},1}, \ldots , x_{\bm{t},p})^{'}$ is a $p$-vector of nonstochastic regressors, $\bm{\beta} = (\beta_1, \ldots, \beta_p)^{'}$ is a vector of unknown regression coefficients and the prime denotes the transposition. Hereafter, it is assumed that $(y_{\bm{t}},X_{\bm{t}})$ is observed on the square sampling domain $P_N =  \{ \bm{t}=(t_1, t_2)^{'} \in \mathbb{Z}^2 | 1 \le t_1 \le N, 1 \le t_2 \le N  \}$. In this case, the sample size is $N^2$. The error terms $\{ \epsilon_{\bm{t}} \}$ follow a stationary random field with mean 0 and spectral density function $f(\bm{\lambda}),\; \bm{\lambda} = (\lambda_1, \lambda_2)^{'} \in [-\pi, \pi]^2$. Thus, the covariance function $\gamma_{\epsilon}(\bm{h})$ of $\{ \epsilon_{\bm{t}} \}$ is given by 
\[
\gamma_{\epsilon}(\bm{h}) = E[\epsilon_{\bm{t}} \epsilon_{\bm{t}+\bm{h}}] = \int_{\Pi^2} \exp(i \bm{h}^{'} \bm{\lambda}) f(\bm{\lambda}) d \bm{\lambda},
\] 
where $\Pi = (-\pi, \pi]$, $\bm{h} = (h_1,h_2)^{'} (\in \mathbb{Z}^2)$ and $\bm{h}^{'} \bm{\lambda} = h_1 \lambda_1 + h_2 \lambda_2$.

When we estimate coefficients in the linear regression model, the GLSE is given by
\[
\hat{\bm{\beta}}_{GLSE} = \left( X^{'} \Sigma^{-1} X \right)^{-1} X^{'} \Sigma^{-1} \bm{y},
\]
where $X = \left( X_{(1,1)}, \ldots, X_{(1,N)},X_{(2,1)}, \ldots, X_{(2,N)}, \ldots, X_{(N,1)}, \ldots, X_{(N,N)} \right)^{'}$, \\
$\bm{y} = \left( y_{(1,1)}, \ldots, y_{(1,N)},y_{(2,1)}, \ldots, y_{(2,N)}, \ldots, y_{(N,1)}, \ldots, y_{(N,N)} \right)^{'}$, $\Sigma = E[\bm{\epsilon}\bm{\epsilon}^{'}]$ and\\
 $\bm{\epsilon} = \left( \epsilon_{(1,1)}, \ldots, \epsilon_{(1,N)},\epsilon_{(2,1)}, \ldots, \epsilon_{(2,N)}, \ldots, \epsilon_{(N,1)}, \ldots, \epsilon_{(N,N)} \right)^{'}$.
 
It is known that the GLSE is the best linear unbiased estimator (BLUE). However, if the true covariance structure is unknown, the GLSE is infeasible. Additionally, the operation count for computing $\Sigma^{-1} \; (N^2\times N^2)$ of $\hat{\bm{\beta}}_{GLSE} \left(= \left( X^{'} \Sigma^{-1} X \right)^{-1} X^{'} \Sigma^{-1} \bm{y} \right)$ is of order $N^6$. Hence, as the sample size increases, the computation becomes impractical. For example, the Walker Lake dataset in Isaaks and Srivastava (1989) consists of two variables measured at 78000 points on a spatial lattice of 260 $\times$ 300 regular grid points and the soil moisture index in Huang et al. (2010) is derived from GIS layers on a 100 $\times$ 100 grid. 

To reduce the computational burden, we consider the approximation of the true covariance function by the product of the covariance functions of the causal autoregressive process of order $P$ (AR($P$)) in time series literature. Thus,  we obtain the following estimator
\[
\hat{\bm{\beta}}_{PBE} = \left( X^{'} \tilde{\Sigma}^{-1} X \right)^{-1} X^{'} \tilde{\Sigma}^{-1} \bm{y},
\]
where $\tilde{\Sigma} = \tilde{\Sigma}_1 \bigotimes  \tilde{\Sigma}_2$ and $\tilde{\Sigma}_i$ is the covariance matrix of causal AR($P_i$) $(i = 1,2)$. This kind of estimator, in which the true covariance matrix is replaced by an  incorrect one, is called the pseudo best estimator. Each element of $\tilde{\Sigma}$ is denoted by the separable covariance function $\gamma(h_1,h_2) = \gamma_1(h_1) \gamma_2(h_2)$ where $\gamma_i(h_i)$ represents the autocovariance function of AR($P_i$)  $(i=1,2)$. The corresponding spectral density functions are denoted by $g(\lambda_1,\lambda_2)$, $g_1(\lambda_1)$ and $g_2(\lambda_2)$ respectively and it follows that $g(\lambda_1,\lambda_2) = g_1(\lambda_1) g_2(\lambda_2)$. $g(\lambda_1,\lambda_2)$ is a kind of an approximation of $f(\lambda_1,\lambda_2)$. 
 From the property of the Kronecker product (Horn and Johnson 1991; page 244), $\tilde{\Sigma}^{-1} = \left( \tilde{\Sigma}_1 \bigotimes  \tilde{\Sigma}_2 \right)^{-1} = \tilde{\Sigma}_1^{-1} \bigotimes  \tilde{\Sigma}_2^{-1}$ and we can obtain the exact form of the inverse of the covariance matrix $\tilde{\Sigma}_i$ ($i=1,2$) given by the  autoregressive process (e.g., Anderson 1971; page 576). This makes it much faster to calculate $\hat{\bm{\beta}}_{PBE}$ by the separable approximation of the true covariance function.

The LSE is another alternative to the GLSE, because it does not require the inversion of the covariance matrix. The LSE is defined as
\[
\hat{\bm{\beta}}_{LSE} = \left( X^{'} X \right)^{-1} X^{'} \bm{y}.
\]
We will compare the accuracy of the GLSE, PBE and LSE by evaluating the asymptotic covariance matrix. 

\section{Asymptotic properties of $\hat{\bm{\beta}}_{GLSE}$, $\hat{\bm{\beta}}_{LSE}$ and $\hat{\bm{\beta}}_{PBE}$}

In this section, we derive the asymptotic covariance matrices of $\hat{\bm{\beta}}_{GLSE}$, $\hat{\bm{\beta}}_{LSE}$ and $\hat{\bm{\beta}}_{PBE}$. First, we introduce some assumptions.

Define
\begin{align*}
a_{ij}^{(N,N)}(h_1,h_2) &= \sum_{t_1=1}^{N-h_1}\sum_{t_2=1}^{N-h_2} x_{(t_1+h_1,t_2+h_2),i}x_{(t_1,t_2),j}, \quad h_1,h_2 \ge 0, \\
&= \sum_{t_1=1-h_1}^{N}\sum_{t_2=1-h_2}^{N} x_{(t_1+h_1,t_2+h_2),i}x_{(t_1,t_2),j}, \quad h_1,h_2 \le 0, \\
&= \sum_{t_1=1-h_1}^{N}\sum_{t_2=1}^{N-h_2} x_{(t_1+h_1,t_2+h_2),i}x_{(t_1,t_2),j}, \quad h_1 \le 0, h_2 \ge 0, \\
&= \sum_{t_1=1}^{N-h_1}\sum_{t_2=1-h_2}^{N} x_{(t_1+h_1,t_2+h_2),i}x_{(t_1,t_2),j}, \quad h_1 \ge 0,h_2 \le 0.
\end{align*}
(a) $a_{ii}^{(N,N)}(0,0) \to \infty$ as $N \to \infty, i = 1,\ldots,p$.\\
(b) $\lim_{N \to \infty} a_{ii}^{(N+h_1,N+h_2)}(0,0)/a_{ii}^{(N,N)}(0,0) = 1$ for every $i$ and $h_1,h_2$, $i= 1, \ldots,p$ and $h_1,h_2 \in \mathbb{Z}$.\\
(c) The limit of 
\[
\gamma_{ij}^{(N,N)}(h_1,h_2) = \frac{a_{ij}^{(N,N)}(h_1,h_2)}{\{a_{ii}^{(N,N)}(0,0)a_{jj}^{(N,N)}(0,0) \}^{1/2}}
\]
as $N \to \infty$ exists for every $i,j$ and $h_1,h_2$, ($i,j= 1, \ldots,p$ and $h_1,h_2 \in \mathbb{Z}$).

Let
\[
\rho_{ij}(h_1,h_2) = \lim_{N \to \infty} \gamma_{ij}^{(N,N)}(h_1,h_2)
\]
and $R(h_1,h_2)$ be the $p \times p$ matrix with $(i,j)$th element $\rho_{ij}(h_1,h_2)$.\\
(d) $R(0,0)$ is nonsingular.\\
(e) 
$\{ \epsilon_{\bm{t}} \}$ is a unilateral moving average process,
\[
\epsilon_{(t_1,t_2)}=\sum_{l=0}^{\infty}\sum_{m=0}^{\infty} \theta_{l,m} \eta_{(t_1-l,t_2-m)},
\]
where $\{ \theta_{l,m} \}$ satisfies $\sum_{l,m} | \theta_{l,m} |^2 < \infty$ with $\theta_{0,0}=1$ and $\{ \eta_{(t_1,t_2)} \}$ is white noise with $Var(\eta_{(t_1,t_2)}) = \sigma_{\eta}^2$.

Moreover if we define $\theta(z_1,z_2) = \sum_{l=0}^{\infty}\sum_{m=0}^{\infty} \theta_{l,m}z_1^{l}z_2^{m}$ and $\phi(z_1,z_2)=\theta(z_1,z_2)^{-1}=\sum_{l=0}^{\infty}\sum_{m=0}^{\infty} \phi_{l,m}z_1^{l}z_2^{m}$, then $\{ \phi_{l,m} \}$ satisfies $\sum_{l,m}| \phi_{l,m} | < \infty$ with $\phi_{0,0}=1$.\\
(f) $f(\bm{\lambda})$ is a positive continuous function in $[-\pi, \pi]^2$.\\
(g) $N \max_{1 \le t_1,t_2 \le N} (x_{(t_1,t_2),i})^2/a_{ii}^{(N,N)}(0,0) \to 0$ as $N \to \infty$ ($i = 1,\ldots,p$).\\
(h) $f(\lambda_1,\lambda_2)=f(-\lambda_1,\lambda_2)=f(\lambda_1,-\lambda_2)=f(-\lambda_1,-\lambda_2)$.

We make brief comments on the assumptions. We can view (a)-(d) as a two-dimensional version of so-called Grenander's conditions on $X_{\bm{t}}$ (Grenander and Rosenblatt 1957; Anderson 1971). Under (a)-(d), there exists a Hermitian matrix function $M(\lambda_1,\lambda_2)$ with positive semidefinite increments such that
\[
R(h_1,h_2)=\int_{\Pi^2} \exp(i(h_1 \lambda_1 + h_2 \lambda_2)) dM(\lambda_1, \lambda_2)
\]
(see Cohen and Francos 2002). Put
\begin{align*}
m_{kl}^{(N,N)}(\lambda_1, \lambda_2) = & \left( \sum_{t_1=1}^{N}\sum_{t_2=1}^{N} x_{(t_1,t_2),k} e^{-i(t_1 \lambda_1 + t_2 \lambda_2)} \right) \left( \sum_{t_1=1}^{N}\sum_{t_2=1}^{N} x_{(t_1,t_2),l} e^{i(t_1 \lambda_1 + t_2 \lambda_2)} \right) \\ 
&\Bigg/ \left( (2\pi)^2 \left(a_{kk}^{(N,N)}(0,0) a_{ll}^{(N,N)}(0,0) \right)^{1/2} \right)
\end{align*}
and define
\begin{align*}
M_{kl}^{(N,N)}(\lambda_1, \lambda_2) = \int_{-\pi}^{\lambda_1} \int_{-\pi}^{\lambda_2} m_{kl}^{(N,N)}(\omega_1, \omega_2) d\omega_1 d\omega_2. 
\end{align*}
Let $M_{kl}(\lambda_1, \lambda_2)$ and $M_{kl}^{(N,N)}(\lambda_1, \lambda_2)$ be $p \times p$ matrices with $(k,l)$th element of $M(\lambda_1, \lambda_2)$ and $M^{(N,N)}(\lambda_1, \lambda_2)$ respectively for $k,l = 1,\ldots,p$. If we regard $M^{(N,N)}(\lambda_1, \lambda_2)$ and $M(\lambda_1, \lambda_2)$ as matrix measures in $\Pi^2$, $R^{(N,N)}(h_1,h_2)=[\gamma_{kl}^{(N,N)}(h_1,h_2)]$ and $R(h_1,h_2)$ are their characteristic functions respectively. Then (c) implies
\[
M^{(N,N)}(\lambda_1, \lambda_2) \stackrel{w}{\to} M(\lambda_1, \lambda_2)
\]
as $N \to \infty$, where $\stackrel{w}{\to}$ means $M^{(N,N)}(\lambda_1, \lambda_2)$ converges weakly to $M(\lambda_1, \lambda_2)$ (see Ibragimov and Rozanov 1978). Consequently, for any continuous bounded function $\phi(\bm{\lambda})$ in $\Pi^2$,
\[
\lim_{N \to \infty} \int_{\Pi^2} \phi(\bm{\lambda}) dM^{(N,N)}(\lambda_1, \lambda_2) =  \int_{\Pi^2} \phi(\bm{\lambda}) dM(\lambda_1, \lambda_2).
\]
In accordance with the time series literature, $M(\lambda_1, \lambda_2)$ is called the regression spectral measure of $X_{\bm{t}}$ (Taniguchi et al. 2008). 

From Yajima and Matsuda (2008), under (e), $\{ \epsilon_{\bm{t}} \}$ has a unilateral AR representation,
\[
\sum_{l=0}^{\infty} \sum_{m=0}^{\infty} \phi_{l,m} \epsilon_{(t_1 -l, t_2 - m)} = \eta_{(t_1,t_2)},
\]
which has a quarter-plane dependence. Bronars and Jansen (1987) applied such a model to unemployment rate fluctuations in the United States. Moreover, (e) holds for a random process with separable covariance functions such as the product of one-dimensional Mat\'ern covariance functions, but a random process with an isotropic Mat\'ern covariance function does not satisfy (e) (Yajima and Matsuda 2008).
The spectral density functions of separable covariance function and isotropic one satisfy (h), which is often referred to as axial symmetry.

Next we derive the asymptotic covariance matrix of the three estimators, that is $\hat{\bm{\beta}}_{GLSE}$, $\hat{\bm{\beta}}_{PBE}$ and $\hat{\bm{\beta}}_{LSE}$.

Define
\[
D_{N^2}=diag( \| \bm{x}_1 \|, \ldots, \| \bm{x}_p \| )
\]
where $\bm{x}_i = \left( x_{(1,1),i}, \ldots, x_{(1,N),i}, x_{(2,1),i}, \ldots, x_{(2,N),i}, \ldots, x_{(N,1),i}, \ldots,  x_{(N,N),i} \right)^{'}$ \\and $\| \bm{x}_i \| = (\sum_{t_1=1}^{N} \sum_{t_2=1}^{N} (x_{(t_1,t_2),i})^2)^{1/2} \; (i = 1,\ldots, p)$.
\par 
\bigskip
\noindent
\textbf{Theorem 1 (Yajima and Matsuda (2008))}.
Under (a)-(g),
\begin{align*}
&\lim_{N \to \infty} D_{N^2} E\left[ (\hat{\bm{\beta}}_{GLSE} - \bm{\beta}) (\hat{\bm{\beta}}_{GLSE} - \bm{\beta})^{'} \right] D_{N^2}\\
&= (2\pi)^2 \left( \int_{\Pi^2} \frac{1}{f(\lambda_1,\lambda_2)} dM(\lambda_1,\lambda_2)  \right)^{-1}.
\end{align*}
\par 
\bigskip
\noindent
\textbf{Theorem 2 (Yajima and Matsuda (2008))}.
Under (a)-(d) and (f),
\begin{align*}
&\lim_{N \to \infty} D_{N^2} E\left[ (\hat{\bm{\beta}}_{LSE} - \bm{\beta}) (\hat{\bm{\beta}}_{LSE} - \bm{\beta})^{'} \right] D_{N^2}\\
&= (2\pi)^2 R(0,0)^{-1} \int_{\Pi^2} f(\lambda_1,\lambda_2) dM(\lambda_1,\lambda_2) R(0,0)^{-1}.
\end{align*}

Moreover, Yajima and Matsuda (2008) gave the following necessary and sufficient conditions for the LSE to be asymptotically efficient relative to the GLSE, which means that the asymptotic covariance matrix of the LSE is identical to that of the GLSE. 
\par 
\bigskip
\noindent
\textbf{Theorem 3 (Yajima and Matsuda (2008))}.
Under (a)-(g), the LSE is asymptotically efficient relative to the GLSE if and only if $M(\lambda_1,\lambda_2)$ increases at not more than $p$ values of $\bm{\lambda}$ ($\in [0,\pi]^2$) 
and the sum of the ranks of the increases in $M(\lambda_1,\lambda_2)$ is $p$.
\par 
\bigskip
Theorems 1-3 can be regarded as an extension of Grenander and Rosenblatt (1957) for $d=1$ to spatial processes. The following theorem is our main result. The proof is given in Appendix B.
\par 
\bigskip
\noindent
\textbf{Theorem 4}.
Under (a)-(d) and (f)-(h),
\begin{align*}
&\lim_{N \to \infty} D_{N^2} E\left[ (\hat{\bm{\beta}}_{PBE} - \bm{\beta}) (\hat{\bm{\beta}}_{PBE} - \bm{\beta})^{'} \right] D_{N^2}\\
&= (2\pi)^2 \left( \int_{\Pi^2} \frac{1}{g(\lambda_1,\lambda_2)} dM(\lambda_1,\lambda_2)  \right)^{-1} \left( \int_{\Pi^2} \frac{f(\lambda_1,\lambda_2)}{g^2(\lambda_1,\lambda_2)} dM(\lambda_1,\lambda_2)  \right) \\
& \times \left( \int_{\Pi^2} \frac{1}{g(\lambda_1,\lambda_2)} dM(\lambda_1,\lambda_2)  \right)^{-1}.
\end{align*}

Note that the conditions required in Theorems 1-4 are not in common. In particular, unlike Theorems 2 and 4, Theorems 1 and 3 do not hold due to (e) if the random process has an isotropic Mat\'ern covariance function. That kind of separable covariance function satisfies (e), (f) and (h) of Theorems 1-4.

Amemiya (1973) and Engle (1974) investigated the asymptotic properties of $\hat{\bm{\beta}}_{PBE}$ for the case $d=1$  and Theorem 4 is an extension of their theoretical results to the spatial case. Additionally, it can be regarded as an extension of Theorem 1 or 2 because the asymptotic covariance matrix in Theorem 4 is identical to that in Theorem 1 or 2 under the appropriate assumptions when $g(\lambda_1, \lambda_2) = f(\lambda_1, \lambda_2)$ or $g(\lambda_1, \lambda_2)$ is a constant.

\section{Computational experiments}

We conduct Monte Carlo simulations using MATLAB. The first experiment examines the convergence and finite sample accuracy of Theorem 4 for different sample sizes, separable approximations and true covariance functions.
We consider linear regression models with one regressor
\[
y_{(t_1,t_2)} = \beta x_{(t_1,t_2)} + \epsilon_{(t_1,t_2)}
\]
for $1 \le t_1,t_2 \le N$ and $\beta = 2$. For the regressor $x_{(t_1,t_2)}$ satisfying (a)-(d) and (g), our computational experiments consider the polynomial trend $x_{(t_1,t_2)} = t_1 t_2$, the harmonic trend $x_{(t_1,t_2)} = \cos((\pi/2)t_1)\cos((\pi/2)t_2)$ and the polynomial plus harmonic trend $x_{(t_1,t_2)} = 1 + \cos((\pi/2)t_1)\cos((\pi/2)t_2)$ (see Toyooka 1985 for the case $d=1$). These regressors are useful for obtaining row and column effects and a kind of periodicity.
 The jumps of $M(\lambda_1,\lambda_2)$ are 1 at $(\lambda_1,\lambda_2)=(0,0)$ and 1/4 at $(\lambda_1,\lambda_2)=(\pi/2,\pi/2),(-\pi/2,\pi/2),(\pi/2,-\pi/2)$,\\$(-\pi/2,-\pi/2)$ for the polynomial and  harmonic trends respectively. For these two regressors, it follows from a routine calculation that the asymptotic variance of the PBE does not depend on separable approximations and the asymptotic variance of the PBE is identical to that of the LSE. As a result, the asymptotic variance of the PBE in each true case is the same for the first and second regressors. In the third regressor, the jumps of $M(\lambda_1,\lambda_2)$ are 4/5 at $(\lambda_1,\lambda_2)=(0,0)$ and 1/20 at $(\lambda_1,\lambda_2)=(\pi/2,\pi/2),(-\pi/2,\pi/2),(\pi/2,-\pi/2),(-\pi/2,-\pi/2)$. These jump values are used to  calculate the asymptotic variance of the GLSE, LSE and PBE. 

For the true covariance function $\gamma_{\epsilon}(\bm{h})$, we consider the following six models. Some models include the Mat\'ern covariance function which is popular in spatial statistics because of its flexibility with  spatial data. The first and second models are
\[
\gamma_{\epsilon}(\bm{h})=c(\|\bm{h}\|) \quad \|\bm{h}\| = \sqrt{h_1^2 + h_2^2},
\]
where 
\[
c(x) = \frac{\sigma^2}{2^{\nu-1}\Gamma(\nu)} \left( \frac{2 \sqrt{\nu}|x|}{\rho} \right)^{\nu} K_{\nu} \left( \frac{2 \sqrt{\nu}|x|}{\rho} \right) 
\]
and $K_{\nu}(\cdot)$ is the modified Bessel function of the second kind of order $\nu$. This is the isotropic Mat\'ern covariance function and we set $\nu = 2, \rho = 3, \sigma^2 = 1$ and $\nu = 1, \rho = 3, \sigma^2 = 1$ respectively. 
The third model is the product of one-dimensional Mat\'ern covariance functions with $\nu = 2, \rho = 3, \sigma^2 = 1$ and $\nu = 1, \rho = 3, \sigma^2 = 1$ respectively. The fourth model is the product of the one-dimensional Mat\'ern covariance function with $\nu = 1, \rho = 3, \sigma^2 = 1$ and the autocovariance function of AR($2$)
\[
c^{*}(x) = \frac{\sigma_{*}^2 \xi_1^2 \xi_2^2}{(\xi_1 \xi_2 -1)( \xi_2 - \xi_1)} \left[ \frac{\xi_1^{1-|x|}}{\xi_1^2 -1} - \frac{\xi_2^{1-|x|}}{\xi_2^2 -1} \right],
\] 
where $\xi_1 = (2/3)(1 + \sqrt{3}i), \xi_2 = (2/3)(1 - \sqrt{3}i)$ and $\sigma_{*}^2$ is chosen such that $c^{*}(0)=1$. It is inappropriate to approximate $c^{*}(x)$ by that of AR($1$) because a successive negative correlation exists. The fifth model is the product of the autocovariance function of AR($2$) in the fourth one and that of AR($1$)
\[
c^{**}(x) = \frac{\sigma_{**}^2}{1-\phi^2} \phi^{|x|},
\]
where $\phi = 0.5$ and $\sigma_{**}^2$ is chosen such that $c^{**}(0)=1$. The sixth model is the product of two autocovariance functions of AR($1$) with $\phi = 0.9$ and a scale parameter such that $\gamma_{\epsilon}(0)=1$. These separable covariance functions are used to check the convergence and finite sample accuracy of Theorem 4 and compare the asymptotic variance of the PBE with that of the GLSE. 
Because all the models satisfy (f) and (h), Theorem 4 holds in these settings.

Next, we explain the method of the separable approximation to obtain the approximated covariance function $\gamma(h_1,h_2)$. This is similar to the Yule-Walker estimator in time series literature. Three types are adopted, namely AR($1$)$\times$AR($1$), AR($1$)$\times$AR($2$) and AR($2$)$\times$AR($2$). The first approximation is expressed by
\begin{align*}
\gamma(h_1,h_2) &= \gamma_1(h_1)\gamma_2(h_2) \\
&= \frac{\sigma_1^2}{1-\phi_1^2} \phi_1^{|h_1|}\frac{\sigma_2^2}{1-\phi_2^2} \phi_2^{|h_2|}.
\end{align*}
Each parameter is given by
\[
\hat{\phi}_1 = \frac{\hat{\gamma}(1,0)}{\hat{\gamma}(0,0)},\;\hat{\phi}_2 = \frac{\hat{\gamma}(0,1)}{\hat{\gamma}(0,0)} \;
 \mbox{and} \; \hat{\sigma}_{12}^2 = \frac{\hat{\gamma}(0,0)}{\gamma^{'} \left(0,0;\hat{\phi}_1,\hat{\phi}_2 \right)},
\]
where 
\begin{align*}
&\hat{\gamma}(\bm{h}) = \frac{1}{N(\bm{h})} \sum_{(i,j) \in S(\bm{h})} (\hat{\epsilon}_{\bm{t}_i} - \bar{\hat{\epsilon}}) (\hat{\epsilon}_{\bm{t}_j} - \bar{\hat{\epsilon}}), \\
&S(\bm{h}) = \{ (i,j) | \bm{h} = \bm{t}_i - \bm{t}_j, \bm{t}_i,\bm{t}_j \in P_N \}, \; N(\bm{h}) = \# \{ S(\bm{h}) \}, \\
&\hat{\epsilon}_{(t_1,t_2)} = y_{(t_1,t_2)} - \hat{\beta}_{LSE} x_{(t_1,t_2)}, \; \bar{\hat{\epsilon}} = \frac{1}{N^2} \sum_{t_1 = 1}^{N} \sum_{t_2 = 1}^{N} \hat{\epsilon}_{(t_1,t_2)} \\
\intertext{and}
& \gamma^{'} \left(0,0;\hat{\phi}_1,\hat{\phi}_2 \right) = \frac{1}{\left(1-\hat{\phi}_1^2 \right) \left(1-\hat{\phi}_2^2 \right)}.
\end{align*}
$\hat{\sigma}_{12}^2$ is an estimator of $\sigma_1^2 \sigma_2^2$ and is necessary for the calculation of the approximated spectral density function $g(\lambda_1,\lambda_2)$. For the second one, the covariance function used in the separable approximation is
\begin{align*}
&\gamma(h_1,h_2) = \gamma_1(h_1)\gamma_2(h_2) \\
& = \frac{\sigma_1^2}{1-\phi_1^2} \phi_1^{|h_1|} \frac{\sigma_2^2 \xi_1^2 \xi_2^2}{(\xi_1 \xi_2 -1)( \xi_2 - \xi_1)} \left[ \frac{\xi_1^{1-|h_2|}}{\xi_1^2 -1} - \frac{\xi_2^{1-|h_2|}}{\xi_2^2 -1} \right],  \\
&\hat{\xi}_1=\frac{\hat{a} + \sqrt{\hat{a}^2 + 4\hat{b}}}{-2\hat{b}}\; \mbox{and} \; \hat{\xi}_2=\frac{\hat{a} - \sqrt{\hat{a}^2 + 4\hat{b}}}{-2\hat{b}},
\end{align*}
where 
\begin{align*}
&
\begin{pmatrix}
  \hat{a}  \\
  \hat{b}
\end{pmatrix}
=
\begin{pmatrix}
  \hat{\rho}_2{(0)} & \hat{\rho}_2{(-1)}  \\
  \hat{\rho}_2{(1)} & \hat{\rho}_2{(0)}
\end{pmatrix}^{-1}
\begin{pmatrix}
  \hat{\rho}_2{(1)}  \\
  \hat{\rho}_2{(2)}
\end{pmatrix}, \\
&\hat{\rho}_2{(0)}=1,\; \hat{\rho}_2{(1)} = \frac{\hat{\gamma}(0,1)}{\hat{\gamma}(0,0)}=\hat{\rho}_2{(-1)}\; \mbox{and} \; \hat{\rho}_2{(2)} = \frac{\hat{\gamma}(0,2)}{\hat{\gamma}(0,0)}.
\end{align*}
$\phi_1$ and $\sigma_1^2 \sigma_2^2$ are given by a similar calculation as in the first approximation model. 
As before, the parameters are given for the case AR($2$)$\times$AR($2$).

We generate the error terms $\{ \epsilon_{(t_1,t_2)} \}_{ 1 \le t_1,t_2 \le N}$ by a multivariate normal distribution with mean $\bm{0}$ and covariance matrix $\Sigma$ and the sample size is $N^2 = 20^2$ or $60^2$. Next, $\hat{\beta}_{LSE}$ is calculated 1000 times for $N = 60$ and we obtain 1000 sets of parameters in $\gamma$ from $\hat{\epsilon}_{(t_1,t_2)} = y_{(t_1,t_2)} - \hat{\beta}_{LSE} x_{(t_1,t_2)}$ by the preceding separable approximation. Using $\gamma$ and $g(\lambda_1,\lambda_2)$ with the average parameter values, we calculate the empirical variance multiplied by $D_{N^2}^2$ from 1000 iterations, as well as the asymptotic variance, for $\hat{\beta}_{PBE}$ and $\hat{\beta}_{LSE}$. Note that the average parameter values calculated from $\hat{\beta}_{LSE}$ for $N = 60$ are used for the case $N = 20$, allowing us to use the fixed $g(\lambda_1,\lambda_2)$ in Theorem 4 throughout Section 4. This improves the accuracy of the PBE in the case $N = 20$, but does not change the conclusion of the computer experiments, that is, the PBE outperforms the LSE when the average parameter value is calculated from $\hat{\beta}_{LSE}$ for $N = 20$. These simulations are omitted from this paper. When calculating the spectral density function of the Mat\'ern covariance function, we must consider the aliasing phenomenon (Fuentes 2005) because the observations exist on an integer lattice.

\begin{table}[htbp]
\begin{center}
\caption{Summary of results from the first experiment for the polynomial trend.}
\label{sim1reg1}
\scalebox{0.75}{
\begin{tabular}{c|c|c|c|c|c|c} \hline \hline
case & \multicolumn{3}{c|}{$\nu$ = 2} & \multicolumn{3}{c}{$\nu$ = 1} \\\hline
approximation & AR(1) $\times$ AR(1) & AR(1) $\times$ AR(2) & AR(2) $\times$ AR(2) & AR(1) $\times$ AR(1) & AR(1) $\times$ AR(2) & AR(2) $\times$ AR(2) \\\hline
N=20 & 18.073 & 16.922 & 16.015 & 16.619 & 16.557 & 16.351 \\
N=60 & 24.288 & 22.934 & 22.085 & 23.315 & 23.283 & 23.334 \\\hline
asymptotic variance & 28.276 & 28.276 & 28.276 & 28.296 & 28.296 & 28.296 \\\hline
case & \multicolumn{3}{c|}{($\nu$ = 2) $\times$ ($\nu$ = 1)} & \multicolumn{3}{c}{($\nu$ = 1) $\times$ AR(2)} \\\hline
approximation & AR(1) $\times$ AR(1) & AR(1) $\times$ AR(2) & AR(2) $\times$ AR(2) & AR(1) $\times$ AR(1) & AR(1) $\times$ AR(2) & AR(2) $\times$ AR(2) \\\hline
N=20 & 16.421 & 16.027 & 14.180 & 3.998 & 3.309 & 3.236 \\
N=60 & 20.474 & 20.002 & 18.524 & 3.669 & 3.686 & 3.664 \\\hline
asymptotic variance & 23.624 & 23.624 & 23.624 & 3.766 & 3.766 & 3.766 \\\hline
case & \multicolumn{3}{c|}{AR(1) $\times$ AR(2)} & \multicolumn{3}{c}{AR(1) $\times$ AR(1)} \\\hline
approximation & AR(1) $\times$ AR(1) & AR(1) $\times$ AR(2) & AR(2) $\times$ AR(2) & AR(1) $\times$ AR(1) & AR(1) $\times$ AR(2) & AR(2) $\times$ AR(2) \\\hline
N=20 & 2.595 & 2.175 & 2.114 & 44.824 & 42.999 & 41.756 \\
N=60 & 2.437 & 2.257 & 2.318 & 151.399 & 154.121 & 159.830 \\\hline
asymptotic variance & 2.392 & 2.392 & 2.392 & 360.999 & 360.999 & 360.999 \\\hline
\end{tabular}
}
\end{center}
\end{table}

\begin{table}[htbp]
\begin{center}
\caption{Summary of results from the first experiment for the harmonic trend.}
\label{sim1reg2}
\scalebox{0.75}{
\begin{tabular}{c|c|c|c|c|c|c} \hline \hline
case & \multicolumn{3}{c|}{$\nu$ = 2} & \multicolumn{3}{c}{$\nu$ = 1} \\\hline
approximation & AR(1) $\times$ AR(1) & AR(1) $\times$ AR(2) & AR(2) $\times$ AR(2) & AR(1) $\times$ AR(1) & AR(1) $\times$ AR(2) & AR(2) $\times$ AR(2) \\\hline
N=20 & 0.101 & 0.101 & 0.102 & 0.240 & 0.242 & 0.244 \\
N=60 & 0.108 & 0.108 & 0.108 & 0.223 & 0.222 & 0.222 \\\hline
asymptotic variance & 0.103 & 0.103 & 0.103 & 0.222 & 0.222 & 0.222 \\\hline
case & \multicolumn{3}{c|}{($\nu$ = 2) $\times$ ($\nu$ = 1)} & \multicolumn{3}{c}{($\nu$ = 1) $\times$ AR(2)} \\\hline
approximation & AR(1) $\times$ AR(1) & AR(1) $\times$ AR(2) & AR(2) $\times$ AR(2) & AR(1) $\times$ AR(1) & AR(1) $\times$ AR(2) & AR(2) $\times$ AR(2) \\\hline
N=20 & 0.061 & 0.061 & 0.059 & 0.233 & 0.213 & 0.211 \\
N=60 & 0.057 & 0.057 & 0.057 & 0.205 & 0.199 & 0.199 \\\hline
asymptotic variance & 0.055 & 0.055 & 0.055 & 0.209 & 0.209 & 0.209 \\\hline
case & \multicolumn{3}{c|}{AR(1) $\times$ AR(2)} & \multicolumn{3}{c}{AR(1) $\times$ AR(1)} \\\hline
approximation & AR(1) $\times$ AR(1) & AR(1) $\times$ AR(2) & AR(2) $\times$ AR(2) & AR(1) $\times$ AR(1) & AR(1) $\times$ AR(2) & AR(2) $\times$ AR(2) \\\hline
N=20 & 0.449 & 0.468  & 0.449  & 0.012 & 0.012 & 0.012 \\
N=60 & 0.422 & 0.438 & 0.40  & 0.011 & 0.012 & 0.011 \\\hline
asymptotic variance & 0.419 & 0.419 & 0.419 & 0.011 & 0.011 & 0.011 \\\hline
\end{tabular}
}
\end{center}
\end{table}

\begin{table}[htbp]
\begin{center}
\caption{Summary of results from the first experiment for the polynomial plus harmonic trend.}
\label{sim1reg3}
\scalebox{0.75}{
\begin{tabular}{c|c|c|c|c|c|c} \hline \hline
case & \multicolumn{3}{c|}{$\nu$ = 2} & \multicolumn{3}{c}{$\nu$ = 1} \\\hline
approximation & AR(1) $\times$ AR(1) & AR(1) $\times$ AR(2) & AR(2) $\times$ AR(2) & AR(1) $\times$ AR(1) & AR(1) $\times$ AR(2) & AR(2) $\times$ AR(2) \\\hline
N=20 & 0.538 & 0.542 & 0.547 & 1.011 & 1.005 & 1.003 \\
N=60 & 0.497 & 0.498 & 0.498 & 1.069 & 1.066 & 1.063 \\\hline
asymptotic variance & 0.516 & 0.515 & 0.513 & 1.097 & 1.095 & 1.088 \\\hline
case & \multicolumn{3}{c|}{($\nu$ = 2) $\times$ ($\nu$ = 1)} & \multicolumn{3}{c}{($\nu$ = 1) $\times$ AR(2)} \\\hline
approximation & AR(1) $\times$ AR(1) & AR(1) $\times$ AR(2) & AR(2) $\times$ AR(2) & AR(1) $\times$ AR(1) & AR(1) $\times$ AR(2) & AR(2) $\times$ AR(2) \\\hline
N=20 & 0.287 & 0.287 & 0.292 & 1.124 & 0.903 & 0.896 \\
N=60 & 0.287 & 0.286 & 0.284 & 0.985 & 0.849 & 0.832 \\\hline
asymptotic variance & 0.271 & 0.271 & 0.270  & 0.980  & 0.856 & 0.856 \\\hline
case & \multicolumn{3}{c|}{AR(1) $\times$ AR(2)} & \multicolumn{3}{c}{AR(1) $\times$ AR(1)} \\\hline
approximation & AR(1) $\times$ AR(1) & AR(1) $\times$ AR(2) & AR(2) $\times$ AR(2) & AR(1) $\times$ AR(1) & AR(1) $\times$ AR(2) & AR(2) $\times$ AR(2) \\\hline
N=20 & 1.630 & 1.108  & 1.288  & 0.057  & 0.058 & 0.062 \\
N=60 & 1.569 & 1.198 & 1.263  & 0.054 & 0.058 & 0.055 \\\hline
asymptotic variance & 1.580  & 1.231 & 1.231 & 0.055 & 0.055 & 0.055 \\\hline
\end{tabular}
}
\end{center}
\end{table}

Tables~\ref{sim1reg1} - \ref{sim1reg3} show that as $N$ increases, the empirical variance multiplied by $D_{N^2}^2$ goes to the asymptotic variance in all cases, as in Theorem 4. However, in some cases with the polynomial trend regressor, the sample size seems to be insufficient. 

Moreover, we can plot the approximated spectral density function $g(\lambda_1,\lambda_2)$ with the average parameter  values given in each case for the polynomial plus harmonic trend to check the fit of $g$ to the true spectral density function $f(\lambda_1,\lambda_2)$. Figures 1-3 show $f$ and $g$ on $[0, \pi]^2$ because $f$ and $g$ are axial symmetry. As the order of the autoregressive process increases, the fit of the spectral density function $g$ to $f$ is better at the points $(\lambda_1,\lambda_2)$ where $M(\lambda_1,\lambda_2)$ has jumps.
 This is consistent with the accuracy of the PBE in Table~\ref{sim1reg3}. 

\begin{figure}[H]
\begin{center}
\scalebox{0.32}{\includegraphics{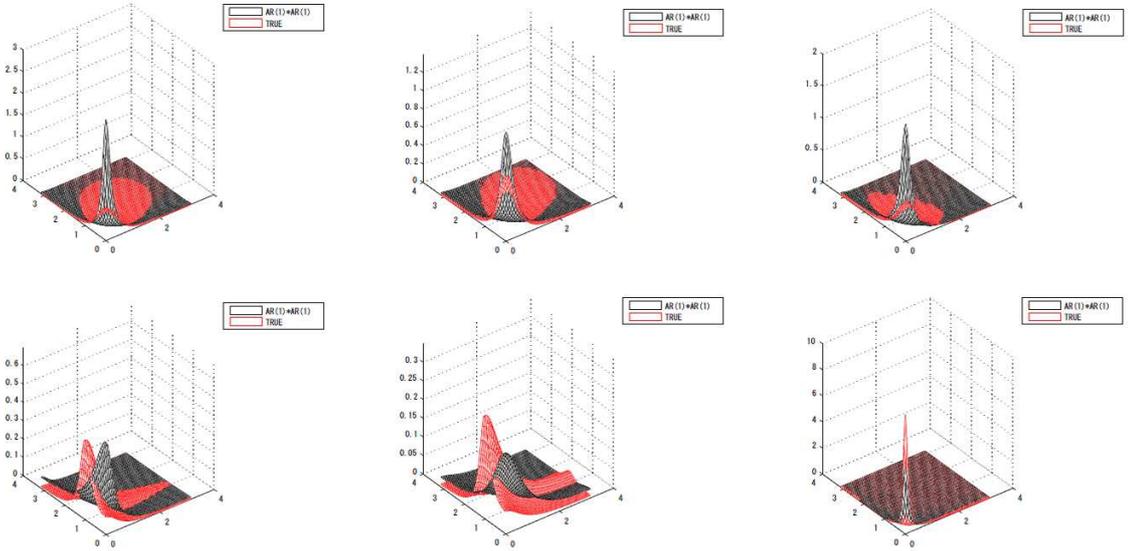}}
  \caption{Approximated spectral density functions $g(\lambda_1,\lambda_2)$ on $[0,\pi]^2$ in the approximation of AR(1)$\times$AR(1). The true covariance function corresponds to the isotropic Mat\'ern ($\nu$ = 2), the isotropic Mat\'ern ($\nu$ = 1) and ($\nu$ = 2)$\times$($\nu$ = 1) for the left, middle and right column in the top row respectively. Similarly, it corresponds to ($\nu$ = 1)$\times$AR(2), AR(1)$\times$AR(2) and AR(1)$\times$AR(1) for the left, middle and right column in the bottom row respectively.}
\end{center}
\end{figure}

\begin{figure}[h]
\begin{center}
\scalebox{0.32}{\includegraphics{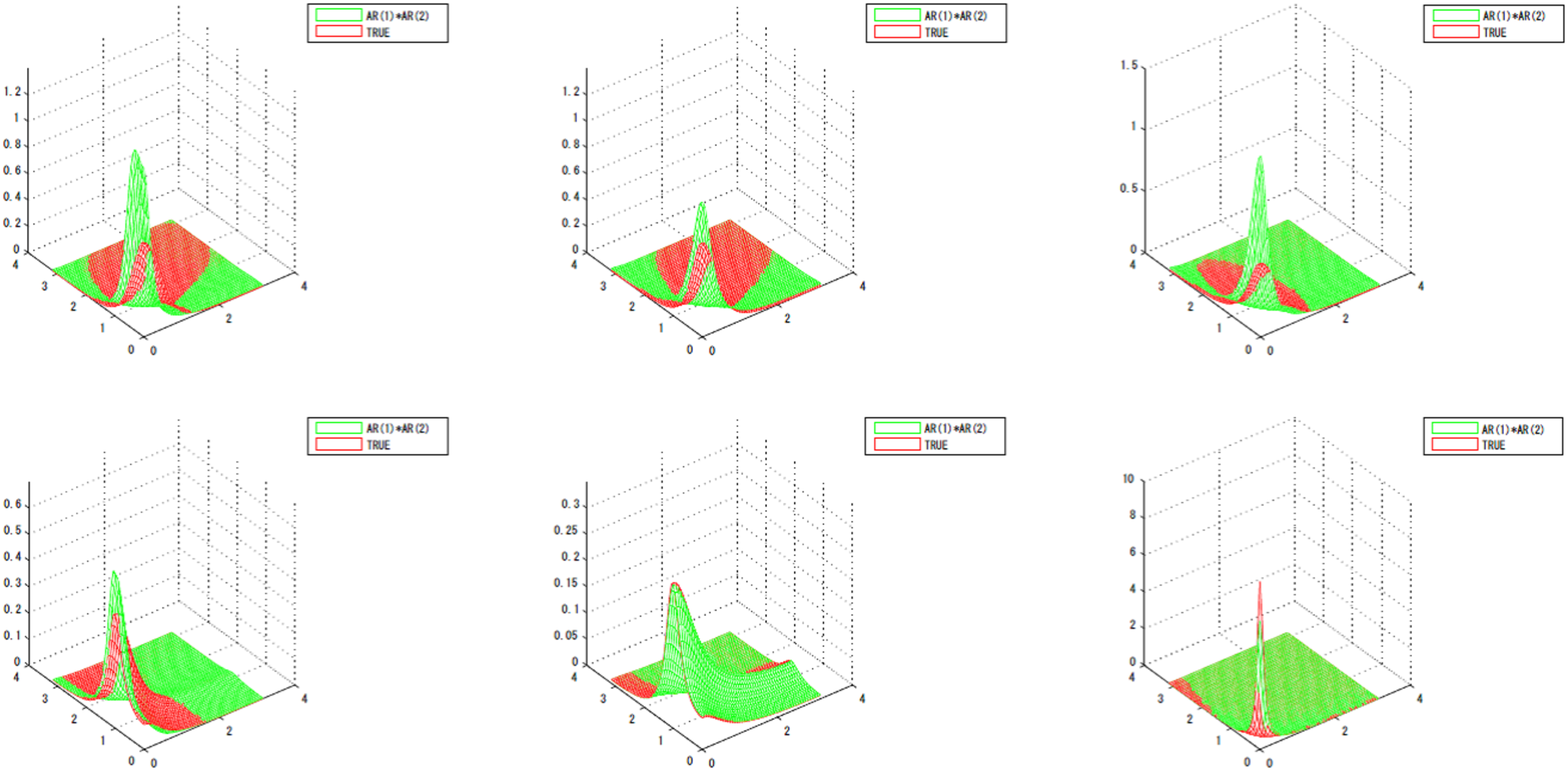}}
  \caption{Approximated spectral density functions $g(\lambda_1,\lambda_2)$ on $[0,\pi]^2$ in the approximation case of AR(1)$\times$AR(2). The true covariance function corresponds to the isotropic Mat\'ern ($\nu$ = 2), the isotropic Mat\'ern ($\nu$ = 1) and ($\nu$ = 2)$\times$($\nu$ = 1) for the left, middle and right column in the top row respectively. Similarly, it corresponds to ($\nu$ = 1)$\times$AR(2), AR(1)$\times$AR(2) and AR(1)$\times$AR(1) for the left, middle and right column in the bottom row respectively.}
\end{center}
\end{figure}

\begin{figure}[h]
\begin{center}
\scalebox{0.32}{\includegraphics{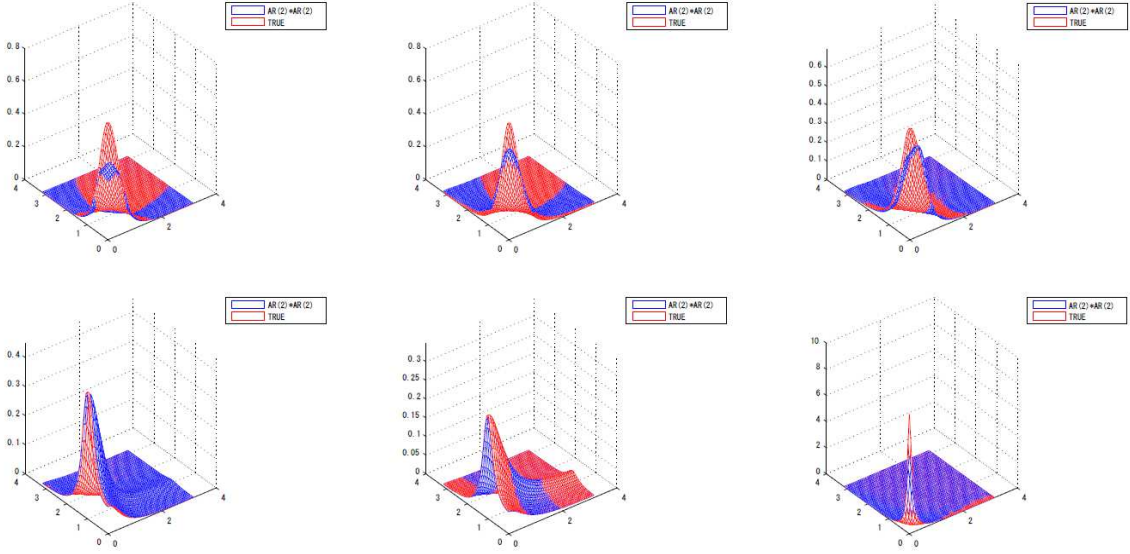}}
  \caption{Approximated spectral density functions $g(\lambda_1,\lambda_2)$ on $[0,\pi]^2$ in the approximation case of AR(2)$\times$AR(2). The true covariance function corresponds to the isotropic Mat\'ern ($\nu$ = 2), the isotropic Mat\'ern ($\nu$ = 1) and ($\nu$ = 2)$\times$($\nu$ = 1) for the left, middle and right column in the top row respectively. Similarly, it corresponds to ($\nu$ = 1)$\times$AR(2), AR(1)$\times$AR(2) and AR(1)$\times$AR(1) for the left, middle and right column in the bottom row respectively.}
\end{center}
\end{figure}

The second experiment compares the finite accuracy and asymptotic variance of the PBE with those of the LSE and the GLSE under the same settings as the first experiment. To see the finite accuracy and asymptotic variance of the three estimators, we calculate the ratios of the empirical variance of the PBE and LSE to that of the GLSE, as well as the theoretical ratios given by the asymptotic variances in Theorems 1, 2 and 4. The empirical variances of the GLSE, LSE and PBE and $g(\lambda_1,\lambda_2)$ are as in the first experiment. The random processes with true covariance functions except the isotropic Mat\'ern covariance function and $(\nu = 1) \times AR(2)$ satisfy (e) and all the models satisfy (f) and (h). Therefore, with the isotropic Mat\'ern covariance function and $(\nu = 1) \times AR(2)$ case, Theorems 1 and 3 do not hold for each regressor. Moreover, for the first and second regressors, it follows from Theorem 3 and a routine calculation that the LSE and PBE are asymptotically efficient except with the isotropic Mat\'ern covariance function and $(\nu = 1) \times AR(2)$ case. Therefore, we show only the empirical ratios in tables for the first and second regressors and the isotropic Mat\'ern covariance function and $(\nu = 1) \times AR(2)$ case for the third regressor. Similarly, from Theorem 3, the LSE is not asymptotically efficient for the third regressor except with the isotropic Mat\'ern covariance function and $(\nu = 1) \times AR(2)$, so that the asymptotic variance of the LSE is different from that of the GLSE. The asymptotic efficiency of the LSE in the isotropic Mat\'ern cases and $(\nu = 1) \times AR(2)$ case is a subject for future work. In subsequent tables, we omit the empirical or asymptotic variances of the LSE, except for the separable approximation AR(1)$\times$AR(1) because they do not depend on the three separable approximations.

\begin{table}
\caption{Summary of results from the second experiment for the polynomial trend.}
\label{sim2reg1}       
\scalebox{0.96}{
\begin{tabular}{cc|c|ccccc|ccccc|ccccc} \hline \hline
 &                                         & case & \multicolumn{5}{c|}{Mat\'ern ($\nu$ = 2)}     &  \multicolumn{5}{c|}{Mat\'ern ($\nu$ = 1)} & \multicolumn{5}{c}{($\nu$ = 2)$\times$($\nu$ = 1)}\\ \hline
 \multicolumn{2}{c|}{approximation}        &  N   & &   LSE  &               &   PBE  &             & &   LSE  &               &  PBE   &          & &   LSE  &           &  PBE    &                       \\ \hline
 \multicolumn{2}{c|}{AR(1) $\times$ AR(1)} &  20  & & 1.328  &               & 1.161  &             & & 1.214  &               & 1.039  &          & & 1.187  &           & 1.159  &                       \\
 &                                         &  60  & & 1.108  &               & 1.116  &             & & 1.137  &               & 1.022  &          & & 1.066  &           & 1.107  &                       \\ \hline
 \multicolumn{2}{c|}{AR(1) $\times$ AR(2)} &  20  & &   -    &               & 1.087  &             & &    -   &               & 1.047  &          & &    -   &           & 1.131  &                       \\
 &                                         &  60  & &   -    &               & 1.054  &             & &    -   &               & 1.021  &          & &    -   &           & 1.081  &                       \\ \hline
 \multicolumn{2}{c|}{AR(2) $\times$ AR(2)} &  20  & &   -    &               & 1.029  &             & &    -   &               & 1.035  &          & &    -   &           & 1.001  &                       \\
 &                                         &  60  & &   -    &               & 1.015  &             & &    -   &               & 1.023  &          & &    -   &           & 1.001  &                       \\ \hline
 &                                         & case & \multicolumn{5}{c|}{($\nu$ = 1)$\times$AR(2)} &  \multicolumn{5}{c|}{AR(1)$\times$AR(2)}     & \multicolumn{5}{c}{AR(1)$\times$AR(1)}       \\ \hline
 \multicolumn{2}{c|}{approximation}        &  N   & &   LSE  &               &   PBE  &             & &   LSE  &               &   PBE  &          & &   LSE  &      &   PBE     &                       \\ \hline
 \multicolumn{2}{c|}{AR(1) $\times$ AR(1)} &  20  & & 1.145  &               & 1.204  &             & & 1.066  &               & 1.172  &          & & 1.929  &           & 1.013   &                       \\
 &                                         &  60  & & 1.073  &               & 1.099  &             & & 1.053  &               & 1.090  &          & & 1.379  &           & 1.002   &                       \\ \hline
 \multicolumn{2}{c|}{AR(1) $\times$ AR(2)} &  20  & &   -    &               & 1.023  &             & &    -   &               & 1.0  &          & &    -   &           &  1.015   &                       \\
 &                                         &  60  & &   -    &               & 1.007  &             & &    -   &               & 1.0   &          & &    -   &           &  1.005   &                       \\ \hline
 \multicolumn{2}{c|}{AR(2) $\times$ AR(2)} &  20  & &   -    &               & 1.001  &             & &    -   &               & 1.004  &          & &    -   &           & 1.012   &                       \\
 &                                         &  60  & &   -    &               & 1.001  &             & &    -   &               & 1.0    &          & &    -   &           & 1.007   &                       \\ \hline
\end{tabular}
}
\end{table}

\begin{table}
\caption{Summary of results from the second experiment for the harmonic trend.}
\label{sim2reg2}       
\scalebox{0.95}{
\begin{tabular}{cc|c|ccccc|ccccc|ccccc} \hline \hline
 &                                         & case & \multicolumn{5}{c|}{Mat\'ern ($\nu$ = 2)}     &  \multicolumn{5}{c|}{Mat\'ern ($\nu$ = 1)} & \multicolumn{5}{c}{($\nu$ = 2)$\times$($\nu$ = 1)}\\ \hline
 \multicolumn{2}{c|}{approximation}        &  N   & &   LSE  &               &   PBE  &             & &   LSE  &               &  PBE   &          & &   LSE  &           &  PBE    &                       \\ \hline
 \multicolumn{2}{c|}{AR(1) $\times$ AR(1)} &  20  & & 1.178  &               & 1.034  &             & & 1.040  &               & 1.028  &          & & 1.483  &           & 1.044  &                       \\
 &                                         &  60  & & 1.064  &               & 1.008  &             & & 1.032  &               & 1.003  &          & & 1.151  &           & 1.016  &                       \\ \hline
 \multicolumn{2}{c|}{AR(1) $\times$ AR(2)} &  20  & &   -    &               & 1.034  &             & &    -   &               & 1.038  &          & &    -   &           & 1.036  &                       \\
 &                                         &  60  & &   -    &               & 1.012  &             & &    -   &               & 1.001  &          & &    -   &           & 1.017  &                       \\ \hline
 \multicolumn{2}{c|}{AR(2) $\times$ AR(2)} &  20  & &   -    &               & 1.042  &             & &    -   &               & 1.047  &          & &    -   &           & 1.011  &                       \\
 &                                         &  60  & &   -    &               & 1.011  &             & &    -   &               & 1.004  &          & &    -   &           & 1.004  &                       \\ \hline
 &                                         & case & \multicolumn{5}{c|}{($\nu$ = 1)$\times$AR(2)} &  \multicolumn{5}{c|}{AR(1)$\times$AR(2)}     & \multicolumn{5}{c}{AR(1)$\times$AR(1)}       \\ \hline
 \multicolumn{2}{c|}{approximation}        &  N   & &   LSE  &               &   PBE  &             & &   LSE  &               &   PBE  &          & &   LSE  &      &   PBE     &                       \\ \hline
 \multicolumn{2}{c|}{AR(1) $\times$ AR(1)} &  20  & & 1.343  &               & 1.103  &             & & 1.203  &               & 1.059  &          & & 1.770  &           & 1.0   &                       \\
 &                                         &  60  & & 1.119  &               & 1.033  &             & & 1.075  &               & 1.032  &          & & 1.271  &           & 1.0   &                       \\ \hline
 \multicolumn{2}{c|}{AR(1) $\times$ AR(2)} &  20  & &   -    &               & 1.008  &             & &    -   &               & 1.0  &          & &    -   &           &  1.0   &                       \\
 &                                         &  60  & &   -    &               & 1.003  &             & &    -   &               & 1.0    &          & &    -   &           &  1.0   &                       \\ \hline
 \multicolumn{2}{c|}{AR(2) $\times$ AR(2)} &  20  & &   -    &               & 1.001  &             & &    -   &               & 1.0    &          & &    -   &           & 1.007   &                       \\
 &                                         &  60  & &   -    &               & 1.0    &             & &    -   &               & 1.0    &          & &    -   &           & 1.005   &                       \\ \hline
\end{tabular}
}
\end{table}

Tables~\ref{sim2reg1} and \ref{sim2reg2} summarize the results for the first and second regressors. As $N$ increases, the finite efficiency goes to 1 based on theoretical results except with the isotropic Mat\'ern covariance function  and $(\nu = 1) \times AR(2)$ case. In the isotropic and $(\nu = 1) \times AR(2)$ cases of this simulation, the finite efficiency is close to 1. When AR(1)$\times$AR(2) and AR(1)$\times$AR(1) are the true covariance functions, the PBE of the corresponding approximation is often superior to those in other cases. Moreover, although the asymptotic variances of the LSE and PBE are equal to each other, the PBE outperforms the LSE in many of the $N = 60$ cases in terms of the ratio of the empirical variances. This is because the PBE includes information about the approximated spatial correlation structure. Throughout these simulations, as the order of the autoregressive process in the separable approximation is larger, the efficiency tends to be often better.

\begin{table}
\caption{Summary of results from the second experiment for the polynomial plus harmonic trend.}
\label{sim2reg3}       
\scalebox{0.8}{
\begin{tabular}{cc|c|ccccc|ccccc|ccccc} \hline
 &                                         & case & \multicolumn{5}{c|}{Mat\'ern ($\nu$ = 2)}     &  \multicolumn{5}{c|}{Mat\'ern ($\nu$ = 1)} & \multicolumn{5}{c}{($\nu$ = 2)$\times$($\nu$ = 1)}\\ \hline
 \multicolumn{2}{c|}{approximation}        &  N   & &   LSE  &               &   PBE  &             & &   LSE  &               &  PBE   &          & &   LSE  &           &  PBE    &                       \\ \hline
 \multicolumn{2}{c|}{AR(1) $\times$ AR(1)} &  20  & & 37.347 &               & 1.045  &             & & 18.060 &               & 1.049  &          & & 54.357 &           & 1.044  &                       \\
 &                                         &  60  & & 42.525 &               & 1.008  &             & & 20.083 &               & 1.007  &          & & 65.797 &           & 1.018  &                       \\ 
 \multicolumn{2}{c|}{} &    & & - &               & -  &             & & - &               & -  &          & & (70.077) &           & (1.004)  &                       \\ \hline
 \multicolumn{2}{c|}{AR(1) $\times$ AR(2)} &  20  & &   -    &               & 1.054  &             & &    -   &               & 1.043  &          & &    -   &           & 1.044  &                       \\
 &                                         &  60  & &   -    &               & 1.009  &             & &    -   &               & 1.005  &          & &    -   &           & 1.014  &                       \\
  \multicolumn{2}{c|}{ } &    & & - &               & -  &             & & - &               & -  &          & & - &           & (1.003)  &                       \\ \hline
 \multicolumn{2}{c|}{AR(2) $\times$ AR(2)} &  20  & &   -    &               & 1.064  &             & &    -   &               & 1.041  &          & &    -   &           & 1.015  &                       \\
 &                                         &  60  & &   -    &               & 1.010  &             & &    -   &               & 1.002  &          & &    -   &           & 1.006  &                       \\ 
   \multicolumn{2}{c|}{ } &    & & - &               & -  &             & & - &               & -  &          & & - &           & (1.002)  &                       \\ \hline
 &                                         & case & \multicolumn{5}{c|}{($\nu$ = 1)$\times$AR(2)} &  \multicolumn{5}{c|}{AR(1)$\times$AR(2)}     & \multicolumn{5}{c}{AR(1)$\times$AR(1)}       \\ \hline
 \multicolumn{2}{c|}{approximation}        &  N   & &   LSE  &               &   PBE  &             & &   LSE  &               &   PBE  &          & &   LSE  &      &   PBE     &                       \\ \hline
 \multicolumn{2}{c|}{AR(1) $\times$ AR(1)} &  20  & & 3.263  &               & 1.259  &             & & 1.661  &               & 1.453  &          & & 1691   &           & 1.003  &                       \\
 &                                         &  60  & & 3.652  &               & 1.191  &             & & 1.660  &               & 1.328  &          & & 3732   &           & 1.0   &                       \\ 
    \multicolumn{2}{c|}{ } &    & & - &               & -  &             & & (1.622) &               & (1.283)  &          & & (5242) &           & (1.0)  &                       \\ \hline
 \multicolumn{2}{c|}{AR(1) $\times$ AR(2)} &  20  & &   -    &               & 1.037  &             & &    -   &               & 1.0  &          & &    -   &           &  1.007   &                       \\
 &                                         &  60  & &   -    &               & 1.026  &             & &    -   &               & 1.0  &          & &    -   &           &  1.0   &                       \\ 
  \multicolumn{2}{c|}{ } &    & & - &               & -  &             & & - &               & (1.0)  &          & & - &           & (1.0)  &                       \\ \hline
 \multicolumn{2}{c|}{AR(2) $\times$ AR(2)} &  20  & &   -    &               & 1.006  &             & &    -   &               & 1.002  &          & &    -   &           & 1.004 &                       \\
 &                                         &  60  & &   -    &               & 1.005  &             & &    -   &               & 1.002  &          & &    -   &           & 1.003   &                       \\ 
  \multicolumn{2}{c|}{ } &    & & - &               & -  &             & & - &               & (1.002)  &          & & - &           & (1.001)  &                       \\\hline
\end{tabular}
}
\end{table}

Table~\ref{sim2reg3} shows the results for the polynomial plus harmonic trend. $(\cdot)$ denotes the theoretical ratio of the asymptotic variances of the LSE and PBE in Theorems 2 and 4 to that of the GLSE in Theorem 1. As $N$ increases, the finite efficiency goes to the value of $(\cdot)$ except in the isotropic Mat\'ern cases and $(\nu = 1) \times AR(2)$ case. Because the LSE is not asymptotically efficient for the third regressor, its performance is poor in this case. However, the efficiency of the PBE is close to 1 in both the empirical and theoretical ratios. In particular, the PBE shows good performance even in the cases of the isotropic Mat\'ern class. In fact, the asymptotic variances of the LSE in Table~\ref{sim2reg3}, which are 22.642 and 22.681 with $\nu = 2$ and $\nu = 1$ respectively, are much larger than those of the PBE in Table~\ref{sim1reg3}. Other properties are similar to those in the preceding simulations.   

Finally we examine the computational time of $\hat{\beta}_{GLSE}$, $\hat{\beta}_{LSE}$ and $\hat{\beta}_{PBE}$. We set $N=100$ and adopt the isotropic Mat\'ern covariance function with $\nu = 2, \rho = 3, \sigma^2 = 1$ for the polynomial plus harmonic trend. All computations are carried out on \rm{Linux  powered 3.33GHz Xeon processor with 8 Gbytes RAM}. From Table~\ref{sim3}, we see that the LSE is computationally efficient, but Table~\ref{sim2reg3} shows it can suffer poor performance. The computation time of the PBE including the estimation procedure is faster than that of the GLSE and the empirical ratios of the PBE relative to the GLSE are close to 1.0 in Table~\ref{sim2reg3}. 

\begin{table}
\caption{Time required to calculate each estimator for the three approximations (seconds).}
\label{sim3}       
\scalebox{0.68}{
\begin{tabular}{ccccccc} \hline
 &                                         & LSE & GLSE    &  PBE(AR(1)$\times$AR(1)) & PBE(AR(1)$\times$AR(2)) &  PBE(AR(2)$\times$AR(2))     \\ \hline
 \multicolumn{2}{c}{estimation procedure (sec.)}  &  -  &  -     &          0.0096          &           0.0086        &            0.013             \\ \hline
 \multicolumn{2}{c}{calculation time of estimators (sec.)} &  $4.43*10^{-4}$  & 384.793  &  2.368   & 2.404  &  2.399                \\ \hline
\end{tabular}
}
\end{table}

\section{Conclusion and future studies}

We have proposed the PBE using the separable approximation of the covariance function on lattice data as an alternative estimator of the GLSE which is practically infeasible owing to its computational burden and the unknown covariance structure. We derived the asymptotic covariance matrix of the PBE and examined the effect of the misspecification of the covariance function in the GLSE. The PBE by the separable approximation works well in many simulations even if the true covariance function is isotropic. In particular, when the LSE is not asymptotically efficient, the PBE exhibits superior performance. Moreover, the PBE substantially reduces the computation time  because of its separable structure. 

In future work, we will present a theoretical comparison of the asymptotic covariance matrices of the PBE and LSE. In addition, the extension of the true process to strongly dependent random fields should be considered.
%
%
\renewcommand{\theequation}{A.\arabic{equation}}
\setcounter{equation}{0}
\bigskip
\\
\textbf{Appendix A : Technical Lemmas}
\bigskip

To derive the asymptotic covariance matrix of $\hat{\bm{\beta}}_{PBE}$, we shall prove the following two lemmas.

\begin{Alemm}
\label{lem1}
Let $(\pi_{1,1}^{m_1,m_2},  \pi_{1,2}^{m_1,m_2}, \ldots, \pi_{1,N}^{m_1,m_2}, \pi_{2,1}^{m_1,m_2}, \ldots, \pi_{N,1}^{m_1,m_2}, \ldots, \pi_{N,N}^{m_1,m_2})^{'}$ be the solution of
\begin{align}
\tilde{\Sigma} 
\begin{pmatrix}
  \pi_{1,1}^{m_1,m_2}  \\
  \vdots    \\
  \pi_{N,N}^{m_1,m_2}
\end{pmatrix}
= 
\begin{pmatrix}
  \gamma(m_1,m_2)  \\
  \gamma(m_1,m_2+1)  \\
  \vdots  \\
  \gamma(m_1,m_2+N-1)  \\
  \gamma(m_1+1,m_2)  \\
  \vdots   \\
  \gamma(m_1+N-1,m_2+N-1)
\end{pmatrix},\quad m_1,m_2 \in \mathbb{Z},  \label{bounded}
\end{align}
where $\tilde{\Sigma}$ and $\gamma$ correspond to those in the definition of $\hat{\bm{\beta}}_{PBE}$ in Section 2.
For fixed $m_1 \ge -N + 1$ and $m_2 \ge -N + 1$, 
\[
\sup_{N} \sum_{i=1}^{N} \sum_{j=1}^{N} | \pi_{i,j}^{m_1,m_2} | < \infty.
\]
\end{Alemm}

\begin{proof}
(\ref{bounded}) is expressed by
\[
\left( \tilde{\Sigma}_1 \bigotimes  \tilde{\Sigma}_2  \right)
\begin{pmatrix}
  \pi_{1,1}^{m_1,m_2}  \\
  \vdots    \\
  \pi_{N,N}^{m_1,m_2}
\end{pmatrix}
= 
\begin{pmatrix}
  \gamma_{1}(m_1)  \\
  \vdots    \\
  \gamma_{1}(m_1+N-1)
\end{pmatrix}
\bigotimes
\begin{pmatrix}
  \gamma_{2}(m_2)  \\
  \vdots    \\
  \gamma_{2}(m_2+N-1)
\end{pmatrix}.
\]
Then, from the property of the Kronecker product (see Horn and Johnson 1991; page 244), 
\[
\begin{pmatrix}
  \pi_{1,1}^{m_1,m_2}  \\
  \vdots    \\
  \pi_{N,N}^{m_1,m_2}
\end{pmatrix}
= \left( \tilde{\Sigma}_1^{-1} 
\begin{pmatrix}
  \gamma_{1}(m_1)  \\
  \vdots    \\
  \gamma_{1}(m_1+N-1)
\end{pmatrix}  \right)
\bigotimes
\left( 
\tilde{\Sigma}_2^{-1}
\begin{pmatrix}
  \gamma_{2}(m_2)  \\
  \vdots    \\
  \gamma_{2}(m_2+N-1)
\end{pmatrix}
 \right).
\]
Now, let $(\pi_{1,k}^{m_k},  \pi_{2,k}^{m_k}, \ldots, \pi_{N,k}^{m_k})^{'}$ be the solution of
\begin{align*}
\tilde{\Sigma}_k
\begin{pmatrix}
  \pi_{1,k}^{m_k}  \\
  \vdots    \\
  \pi_{N,k}^{m_k}
\end{pmatrix}
= 
\begin{pmatrix}
  \gamma_{k}(m_k)  \\
  \vdots    \\
  \gamma_{k}(m_k+N-1)
\end{pmatrix}
\end{align*}
for $k=1,2$.  We will prove 
\begin{align}
\sup_N \sum_{i=1}^{N} |\pi_{i,k}^{m_k}| < \infty \quad (m_k \ge -N+1)
\label{bounded2}
\end{align}
by mathematical induction. In $-N+1 \le m_k \le 0$ ($k=1,2$), $\pi_{i,k}^{m_k}=\delta_{i,(-m_k + 1)}$ from the uniqueness of $(\pi_{1,k}^{m_k},  \pi_{2,k}^{m_k}, \ldots, \pi_{N,k}^{m_k})^{'}$ where $\delta_{i,(-m_k + 1)}$ is 1 if $i=-m_k + 1$, otherwise 0. Therefore, (\ref{bounded2}) is clear in this case. When $m_k=1$, 
\[
\tilde{\Sigma}_k
\begin{pmatrix}
  \pi_{1,k}^{1}  \\
  \vdots    \\
  \pi_{N,k}^{1}
\end{pmatrix}
= 
\begin{pmatrix}
  \gamma_{k}(1)  \\
  \vdots    \\
  \gamma_{k}(N)
\end{pmatrix}.
\]
If $z_t$ is AR($P_k$) with the autocovariance function $\gamma_k$, that is $z_t = \phi_{1,k}z_{t-1} + \cdots + \phi_{P_k,k}z_{t-P_k} + \eta_t$ where $\eta_t$ is white noise, $\pi_{j,k}^1 = \phi_{j,k}$ ($j=1,\ldots,P_k$) and $\pi_{j,k}^1 = 0$ ($j > P_k$). Therefore, $\sup_N \sum_{i=1}^{N} |\pi_{i,k}^{1}| < \infty$. Next, assume that (\ref{bounded2}) holds for $m_k=m$. Consider $m_k = m+1$. Then, from a routine calculation, $P_{\overline{\rm{sp}}(z_1,\ldots,z_N)}z_{N+m+1}=\sum_{l=1}^{N} \pi_{l,k}^{m+1} z_{N+1-l}$ where $P_{\overline{\rm{sp}}(z_1,\ldots,z_N)}$ is the projection to the closed subspace $\overline{\rm{sp}}(z_1,\ldots,z_N)$ spanned by $(z_1,\ldots,z_N)$. Moreover, from a similar calculation, 
\begin{align*}
P_{\overline{\rm{sp}}(z_1,\ldots,z_N)} z_{N+m+1} &= P_{\overline{\rm{sp}}(z_1,\ldots,z_N)} \left( P_{\overline{\rm{sp}}(z_1,\ldots,z_{N+1})} z_{N+m+1} \right) \\
&= P_{\overline{\rm{sp}}(z_1,\ldots,z_N)} \left( \sum_{l=1}^{N+1} \pi_{l,k}^{m} z_{N+2-l} \right) \\
&= P_{\overline{\rm{sp}}(z_1,\ldots,z_N)} \pi_{1,k}^{m} z_{N+1} + \sum_{l=2}^{N+1} \pi_{l,k}^{m} z_{N+2-l} \\
&=  \pi_{1,k}^{m} \left( \sum_{l=1}^{N} \pi_{l,k}^{1} z_{N+1-l} \right) + \sum_{l=2}^{N+1} \pi_{l,k}^{m} z_{N+2-l}.
\end{align*}
Therefore, $\pi_{l,k}^{m+1}=\pi_{1,k}^{m} \pi_{l,k}^{1} + \pi_{l+1,k}^{m}$ ($l=1,\ldots,N$). In this case,
\[
\sup_N \sum_{i=1}^{N} |\pi_{i,k}^{m+1}| \le \sup_N \sum_{i=1}^{N} | \pi_{1,k}^{m} | |\pi_{i,k}^{1}| + \sup_N \sum_{i=1}^{N} |\pi_{i+1,k}^{m}| < \infty.
\]
(\ref{bounded2}) holds for $m_k=m+1$. Note that $\pi_{i,j}^{m_1,m_2} = \pi_{i,1}^{m_1}\pi_{j,2}^{m_2}$ for $1 \le i,j \le N$. Then,
\[
\sum_{i=1}^{N} \sum_{j=1}^{N} | \pi_{i,j}^{m_1,m_2} | = \sum_{i=1}^{N} |\pi_{i,1}^{m_1}| \sum_{j=1}^{N} |\pi_{j,2}^{m_2}|
\]
is bounded for $N$.
\end{proof}

\begin{Alemm}
\label{lem2}
Suppose that $h(\lambda_1,\lambda_2)$ is a continuous function on $[-\pi,\pi]^2$ and $h(\lambda_1,\lambda_2)=h(-\lambda_1,\lambda_2)=h(\lambda_1,-\lambda_2)=h(-\lambda_1,-\lambda_2)$. Then, for any sufficiently small $\epsilon > 0$, there exist $h_L(\lambda_1,\lambda_2)$ and $h_U(\lambda_1,\lambda_2)$ such that 
\begin{align*}
& h_L(\lambda_1,\lambda_2) = \sum_{k_1=-K_1}^{K_1} \sum_{k_2=-K_2}^{K_2} a_{k_1,k_2} e^{i(k_1 \lambda_1 + k_2 \lambda_2)}, \quad a_{k_1,k_2} = a_{-k_1,-k_2} \in \mathbb{R}, \\
& h_U(\lambda_1,\lambda_2) = \sum_{k_1=-K_1}^{K_1} \sum_{k_2=-K_2}^{K_2} b_{k_1,k_2} e^{i(k_1 \lambda_1 + k_2 \lambda_2)}, \quad b_{k_1,k_2} = b_{-k_1,-k_2} \in \mathbb{R}, \\
& h_L(\lambda_1,\lambda_2) \le h(\lambda_1,\lambda_2) \le h_U(\lambda_1,\lambda_2) 
 \intertext{and} 
& h_U(\lambda_1,\lambda_2) - h_L(\lambda_1,\lambda_2) \le \epsilon, \quad (\lambda_1,\lambda_2) \in [-\pi,\pi]^2.
\end{align*}
\end{Alemm}

\begin{proof}
We set $D_{n_1, n_2}(y_1,y_2) = \sum_{|j_1| \le n_1}\sum_{|j_2| \le n_2} e^{i(j_1 y_1 + j_2 y_2)} = D_{n_1}(y_1)D_{n_2}(y_2)$ where $D_{n_i}(y_i) = \sum_{|j_i| \le n_i} e^{ij_i y_i}$ ($i=1,2$). $D_{n_i}(y_i)$ is the Dirichlet kernel,
\[
D_{n_i}(y_i) = \begin{cases}
                   \frac{\sin \left[ \left(n_i + \frac{1}{2} \right) y_i \right]}{\sin \left( \frac{1}{2} y_i \right)} & \text{if $y_i \neq 0$}, \\
                   2 n_i + 1 & \text{if $y_i = 0$}.
               \end{cases}
\]
Additionally, consider
\begin{align*}
S_{n_1, n_2} h(\lambda_1^{'},\lambda_2^{'}) &= \sum_{|j_1| \le n_1}\sum_{|j_2| \le n_2} \left( \frac{1}{(2 \pi)^2}
\int_{-\pi}^{\pi} \int_{-\pi}^{\pi} h(\lambda_1,\lambda_2) 
e^{-i(j_1 \lambda_1 + j_2 \lambda_2)} d\lambda_1 d\lambda_2   \right) e_{j_1,j_2} \\
&= \frac{1}{(2 \pi)^2} \int_{-\pi}^{\pi} \int_{-\pi}^{\pi} h(\lambda_1,\lambda_2) D_{n_1}(\lambda_1^{'}-\lambda_1)  D_{n_2}(\lambda_2^{'}-\lambda_2) d\lambda_1 d\lambda_2,
\end{align*}
where $e_{j_1,j_2}=e^{i(j_1 \lambda_1 + j_2 \lambda_2)}$. By defining $h(\lambda_1,\lambda_2) = h(\lambda_1+2\pi n_1^{'},\lambda_2+2\pi n_2^{'})$ ($n_1^{'},n_2^{'} \in \mathbb{Z}$), this can be rewritten as
\[
S_{n_1, n_2} h(\lambda_1^{'},\lambda_2^{'}) = \frac{1}{(2 \pi)^2} \int_{-\pi}^{\pi} \int_{-\pi}^{\pi} h(\lambda_1^{'}-\lambda_1,\lambda_2^{'}-\lambda_2) D_{n_1}(\lambda_1 ) D_{n_2}(\lambda_2) d\lambda_1 d\lambda_2.
\]
Moreover,
\[
\frac{1}{n_1 n_2} \sum_{i_1=0}^{n_1-1} \sum_{i_2=0}^{n_2-1} S_{i_1, i_2} h(\lambda_1^{'},\lambda_2^{'}) = \int_{-\pi}^{\pi} \int_{-\pi}^{\pi} h(\lambda_1^{'}-\lambda_1,\lambda_2^{'}-\lambda_2) K_{n_1}(\lambda_1 ) K_{n_2}(\lambda_2) d\lambda_1 d\lambda_2,
\]
where $K_{n_i}(y_i)$ is defined by
\[
\frac{1}{2\pi n_i} \sum_{j_i=0}^{n_i-1} D_{j_i}(y_i) \quad (i=1,2)
\]
and is called the Fejer kernel.
Finally, it follows from the argument extended from Theorem 2.11.1 of Brockwell and Davis (1991; page 69) and the axial symmetry of $h(\lambda_1,\lambda_2)$ that for any $\epsilon > 0$, 
\[
\left| h(\lambda_1, \lambda_2) - \frac{1}{n_1 n_2} \sum_{i_1=0}^{n_1-1} \sum_{i_2=0}^{n_2-1} S_{i_1, i_2} h(\lambda_1,\lambda_2) \right| < \epsilon,
\]
uniformly on $[-\pi,\pi]^2$ for sufficiently large $n_1$ and $n_2$. In that case,
\begin{align*}
\frac{1}{n_1 n_2} \sum_{i_1=0}^{n_1-1} \sum_{i_2=0}^{n_2-1} S_{i_1, i_2} h(\lambda_1,\lambda_2) = \sum_{j_1=-(n_1-1)}^{n_1-1} \sum_{j_2=-(n_2-1)}^{n_2-1} c_{j_1,j_2} e_{j_1,j_2},
\end{align*}
where $c_{j_1,j_2}=c_{-j_1,j_2}=c_{j_1,-j_2}=c_{-j_1,-j_2}$ from the axial symmetry of $h(\lambda_1,\lambda_2)$. Therefore, $c_{j_1,j_2}$'s are real and
\[
\sum_{j_1=-(n_1-1)}^{n_1-1} \sum_{j_2=-(n_2-1)}^{n_2-1} c_{j_1,j_2} e_{j_1,j_2} = \sum_{j_1=-(n_1-1)}^{n_1-1} \sum_{j_2=-(n_2-1)}^{n_2-1} c_{j_1,j_2}(\cos(j_1 \lambda_1 + j_2 \lambda_2) + i \sin(j_1 \lambda_1 + j_2 \lambda_2)).
\]
Now, it follows that
\[
 \sum_{j_1=-(n_1-1)}^{n_1-1} \sum_{j_2=-(n_2-1)}^{n_2-1} c_{j_1,j_2} \sin(j_1 \lambda_1 + j_2 \lambda_2) = 0.
\]
Therefore, $\sum_{|j_1| \le n_1 - 1} \sum_{|j_2| \le n_2 - 1} c_{j_1,j_2} e_{j_1,j_2}$ is real. For any $\epsilon > 0$, 
\[
\sum_{j_1=-(n_1-1)}^{n_1-1} \sum_{j_2=-(n_2-1)}^{n_2-1} c_{j_1,j_2} e_{j_1,j_2} -\frac{\epsilon}{2} \le h(\lambda_1, \lambda_2) \le \sum_{j_1=-(n_1-1)}^{n_1-1} \sum_{j_2=-(n_2-1)}^{n_2-1} c_{j_1,j_2} e_{j_1,j_2} + \frac{\epsilon}{2}
\]
for sufficiently large $n_1$ and $n_2$. Therefore, setting $K_1 = n_1 - 1, K_2 = n_2 - 1, a_{k_1,k_2}= b_{k_1,k_2} = c_{k_1,k_2}$ except $a_{0,0}=c_{0,0}-\epsilon/2, b_{0,0}=c_{0,0}+\epsilon/2$, the proof is completed.

\end{proof}
\renewcommand{\theequation}{B.\arabic{equation}}
\setcounter{equation}{0}
\noindent
\textbf{Appendix B : Proof of Theorem 4}
\bigskip

We prove Theorem 4 by extending the arguments of Anderson (1971) and Yajima (1994).
\par
\bigskip
\noindent
\textbf{Proof of Theorem 4}
\bigskip
\begin{align*}
&D_{N^2} E\left[ (\hat{\bm{\beta}}_{PBE} - \bm{\beta}) (\hat{\bm{\beta}}_{PBE} - \bm{\beta})^{'} \right] D_{N^2}\\
&= \left( D_{N^2}^{-1} X^{'} \tilde{\Sigma}^{-1} X D_{N^2}^{-1} \right)^{-1} \left( D_{N^2}^{-1} X^{'} \tilde{\Sigma}^{-1} \Sigma \tilde{\Sigma}^{-1} X D_{N^2}^{-1} \right) \left( D_{N^2}^{-1} X^{'} \tilde{\Sigma}^{-1} X D_{N^2}^{-1} \right)^{-1}.
\end{align*}
From Theorem 1, the first and third terms converge to 
\[
\left( D_{N^2}^{-1} X^{'} \tilde{\Sigma}^{-1} X D_{N^2}^{-1} \right)^{-1} \to \left( \frac{1}{(2\pi)^2} \int_{\Pi^2} \frac{1}{g(\lambda_1,\lambda_2)} dM(\lambda_1,\lambda_2) \right)^{-1}
\] 
as $N \to \infty$. Here, we put 
\begin{align*}
f(\lambda_1,\lambda_2) &= \frac{f(\lambda_1,\lambda_2)}{g(\lambda_1,\lambda_2)} \times g(\lambda_1,\lambda_2) \\
&= h(\lambda_1,\lambda_2) \times g(\lambda_1,\lambda_2).
\end{align*}
First, consider the case of 
\[
h(\lambda_1,\lambda_2) = \sum_{l_1=-L_1}^{L_1} \sum_{l_2=-L_2}^{L_2} b_{l_1,l_2} e^{i(l_1 \lambda_1 + l_2 \lambda_2)}, \quad b_{l_1,l_2} = b_{-l_1,-l_2}.
\]
In this case, 
\begin{align*}
&E[\epsilon_{(j_1,j_2)} \epsilon_{(k_1,k_2)}] = \gamma_{\epsilon} (j_1-k_1,j_2-k_2) \\
&= \int_{-\pi}^{\pi} \int_{-\pi}^{\pi} e^{i( (j_1 - k_1)\lambda_1 + (j_2 - k_2)\lambda_2 )} f(\lambda_1,\lambda_2) d\lambda_1 d\lambda_2 \\
&= \sum_{l_1=-L_1}^{L_1} \sum_{l_2=-L_2}^{L_2} b_{l_1,l_2} \int_{-\pi}^{\pi} \int_{-\pi}^{\pi} e^{i( (j_1 - k_1+l_1)\lambda_1 + (j_2 - k_2+l_2)\lambda_2 )} g(\lambda_1,\lambda_2) d\lambda_1 d\lambda_2.
\end{align*} 
Therefore,
\[
\Sigma = \sum_{l_1=-L_1}^{L_1} \sum_{l_2=-L_2}^{L_2} b_{l_1,l_2} \tilde{\Sigma}(l_1,l_2),
\]
where
\begin{align*}
&\tilde{\Sigma}(l_1,l_2)= \\
&\left(
{\arraycolsep = 0.5mm
\begin{array}{cccc}
\gamma(l_1,l_2) & \gamma(l_1,l_2-1) & \cdots & \gamma(l_1+1-N,l_2)  \\
\gamma(l_1,l_2+1) & \gamma(l_1,l_2) & \cdots & \gamma(l_1+1-N,l_2+1)  \\
      \vdots      &         \vdots         &        &      \vdots      \\
\gamma(l_1+N-1,l_2) & \gamma(l_1+N-1,l_2-1) & \cdots & \gamma(l_1,l_2)  \\ 
       \vdots      &            \vdots      &        &     \vdots        \\
\gamma(l_1+N-1,l_2+N-1) & \gamma(l_1+N-1,l_2+N-2) & \cdots & \gamma(l_1,l_2+N-1) 
\end{array}}
\right.
\\
&
\left.
\begin{array}{ccc}
 \cdots & \gamma(l_1+1-N,l_2+2-N) & \gamma(l_1+1-N,l_2+1-N) \\
 \cdots & \gamma(l_1+1-N,l_2+3-N) & \gamma(l_1+1-N,l_2+2-N) \\
                       & \vdots &        \vdots            \\
 \cdots & \gamma(l_1,l_2+2-N) & \gamma(l_1,l_2+1-N) \\ 
                       & \vdots &        \vdots            \\
 \cdots & \gamma(l_1,l_2+1) & \gamma(l_1,l_2) 
\end{array}
\right).
\end{align*}
Consider the case of $l_1 \ge 0$ and $l_2 \ge 0$ for sufficiently large $N$.
In this case, the $(m,n)$th element $\left( D_{N^2}^{-1} X^{'} \tilde{\Sigma}^{-1} \tilde{\Sigma}(l_1,l_2) \tilde{\Sigma}^{-1} X D_{N^2}^{-1} \right)_{m,n}$ is 
\begin{align}
&\frac{1}{\left( a_{mm}^{(N,N)}(0,0)a_{nn}^{(N,N)}(0,0) \right)^{1/2}}  \Bigg\{  (\bm{x}_{(1,\cdot),m}^{'}, \ldots, \bm{x}_{(N,\cdot),m}^{'}) \tilde{\Sigma}^{-1} \nonumber  \\
&\times \left( \right.
  x_{(l_1+1,l_2+1),n},
  \ldots,
  x_{(l_1+1,N),n},
  0  ,
  \ldots   ,
  0   ,
  \ldots   ,
  x_{(N,l_2+1),n}  ,
  \ldots  ,
  x_{(N,N),n}  ,
  0  ,
  \ldots   ,
  0 , \nonumber  \\
&  0   ,
  \ldots   ,
  0
\left. \right)^{'} +  ( \bm{x}_{(1,\cdot ),m}^{'}, \ldots, \bm{x}_{(N,\cdot ),m}^{'} ) \tilde{\Sigma}^{-1} \nonumber \\
& \times  
\Bigg( 
  0  ,
  \ldots  ,
  0  ,
  \sum_{i=1}^{N} \sum_{j=1}^{N} \pi_{i,j}^{l_1+1-N,1} x_{(N+1-i,N+1-j),n}  ,
  \ldots   ,
  \sum_{i=1}^{N} \sum_{j=1}^{N} \pi_{i,j}^{l_1+1-N,l_2} \nonumber \\
&\times x_{(N+1-i,N+1-j),n}   ,
  \ldots   ,
  0  ,
  \ldots  ,
  0  ,
  \sum_{i=1}^{N} \sum_{j=1}^{N} \pi_{i,j}^{0,1} x_{(N+1-i,N+1-j),n}  ,
  \ldots  ,
  \sum_{i=1}^{N} \sum_{j=1}^{N} \pi_{i,j}^{0,l_2} \nonumber \\
&\times x_{(N+1-i,N+1-j),n}   ,
  \sum_{i=1}^{N} \sum_{j=1}^{N} \pi_{i,j}^{1,l_2+1-N} x_{(N+1-i,N+1-j),n}   ,
  \ldots   , \nonumber \\
&  \sum_{i=1}^{N} \sum_{j=1}^{N} \pi_{i,j}^{l_1,l_2} x_{(N+1-i,N+1-j),n}
\Bigg) ^{'}
  \Bigg\} \label{original}
\end{align}
where $\bm{x}_{(i,\cdot),m} = (x_{(i,1),m}, \ldots, x_{(i,N),m})^{'}$ ($i = 1,\ldots,N$) and $\{\pi_{i,j}^{m_1,m_2}\}$ is defined in Lemma~\ref{lem1}. 

Now, $g_j(\lambda_j)$ is expressed by
\[
\frac{\sigma_j^2}{2\pi} \frac{1}{|\phi_j(e^{- i \lambda_j})|^2},
\]
where $\phi_j(z) = 1 - \phi_{j,1}z - \cdots - \phi_{j,P_j}z^{P_j}$ and $\phi_j(z) \neq 0$ for 
$| z | \le 1$ ($j=1,2$). From Anderson (1971; page 576), for $i = 1,2$, $\tilde{\Sigma}^{-1}_i= B_i'B_i/\sigma_i^2$ where
\begin{align*}
B_i = \left(
\begin{array}{cccccccccccc}
b_{i,11} & 0 &  &  & &  & & & \cdots & & & 0 \\
\vdots & \ddots & & &  & &&&&&& \vdots \\
b_{i,P_i 1} & \cdots & b_{i,P_i P_i} & 0 & &&& & \cdots &&&0\\
-\phi_{i,P_i} & \cdots & -\phi_{i,1} & -\phi_{i,0} & 0 &&& & \cdots &&&0\\
0 & -\phi_{i,P_i} & \cdots & -\phi_{i,1} & -\phi_{i,0} & 0&& & \cdots &&&0 \\
&&&&&&&&&&& \\
\vdots &&&&&&&& \ddots &&& \vdots \\
&&&&&&&&&&& \\
&&&&&&&&&&& \\
0&&&&&&&&&& -\phi_{i,1} & -\phi_{i,0}
\end{array}
\right),
\end{align*}
$\phi_{i,0}=-1$ and $b_{11},b_{21},b_{22},\ldots,b_{P_i P_i}$ are chosen so that $\tilde{\Sigma}^{-1}_i= B_i'B_i/\sigma_i^2$. Moreover, from the properties of the Kronecker product (Horn and Johnson 1991; page 244), $\tilde{\Sigma}^{-1} = (\tilde{\Sigma}_1 \bigotimes \tilde{\Sigma}_2)^{-1} = \tilde{\Sigma}_1^{-1} \bigotimes \tilde{\Sigma}_2^{-1}
= (B_1'B_1 \bigotimes B_2'B_2)/(\sigma_1^2 \sigma_1^2) = (B_1' \bigotimes B_2')(B_1 \bigotimes B_2)/(\sigma_1^2 \sigma_1^2) = (B_1 \bigotimes B_2)'(B_1 \bigotimes B_2)/(\sigma_1^2 \sigma_1^2) $. By substituting this expression into the first term of (\ref{original}),
\begin{align}
&(\mbox{the first term of (\ref{original})}) = \frac{1}{\sigma_1^2 \sigma_2^2 \left( a_{mm}^{(N,N)}(0,0)a_{nn}^{(N,N)}(0,0) \right)^{1/2}} \nonumber\\
&\times \Biggl[  \sum_{j_1 = 1}^{P_1} \sum_{i_1 = 1}^{j_1} \sum_{i_2 = 1}^{j_1} b_{1,j_1 i_1} b_{1,j_1 i_2}       
\Bigg\{ \sum_{k=1}^{P_2} \left( \sum_{\alpha = 1}^{k} b_{2,k \alpha} x_{(i_1,\alpha),m} \right) \left( \sum_{\beta = 1}^{k} b_{2,k \beta} x_{(i_2 + l_1,\beta + l_2),n} \right) \nonumber\\
&+ \sum_{k=1}^{N-P_2} \left( \sum_{\alpha = 0}^{P_2} (-\phi_{2,P_2 - \alpha}) x_{(i_1,k + \alpha),m} \right) \left( \sum_{\beta = 0}^{P_2} (-\phi_{2,P_2 - \beta}) x_{(i_2 + l_1,k + l_2 + \beta ),n} 1(k + l_2 + \beta \le N) \right)   \Bigg\} \nonumber\\
&+ \sum_{a = 1}^{N-P_1} \sum_{i_1 = 0}^{P_1} \sum_{i_2 = 0}^{P_1} (-\phi_{1,P_1 - i_1})(-\phi_{1,P_1 - i_2})1(l_1 + a + i_2 \le N) \Bigg\{ \sum_{k=1}^{P_2} \left( \sum_{\alpha = 1}^{k} b_{2,k \alpha} x_{(i_1 + a, \alpha),m} \right) \nonumber\\
&\times \left( \sum_{\beta = 1}^{k} b_{2,k \beta} x_{(l_1 + a + i_2, N - l_2),n} \right) + \sum_{k=1}^{N-P_2} 
\left( \sum_{\alpha = 0}^{P_2} (-\phi_{2,P_2 - \alpha}) x_{(i_1 + a, k + \alpha),m} \right) \nonumber\\
&\times \left( \sum_{\beta = 0}^{P_2} (-\phi_{2,P_2 - \beta}) x_{(l_1 + a + i_2, k + l_2 + \beta ),n} 1(k + l_2 + \beta \le N) \right) 
\Bigg\}
\Biggl].  \label{first} 
\end{align}
It follows from (a)-(d) and a routine calculation that (\ref{first}) converges to 
\begin{align*}
&\frac{1}{\sigma_1^2 \sigma_2^2} \sum_{i_1=0}^{P_1} \sum_{i_2=0}^{P_1} (-\phi_{1,P_1-i_1}) (-\phi_{1,P_1-i_2}) \sum_{\alpha =0}^{P_2} \sum_{\beta =0}^{P_2}
(-\phi_{2,P_2-\alpha}) (-\phi_{2,P_2-\beta})  \\
&\times \int_{\Pi^2} e^{i(-l_1 - i_2 + i_1) \lambda_1}  e^{i(-l_2 - \beta + \alpha) \lambda_2} d M_{mn}(\lambda_1,\lambda_2) \\
= & \frac{1}{(2\pi)^2} \int_{\Pi^2} \frac{e^{-i(l_1 \lambda_1 + l_2 \lambda_2)}}{g_1(\lambda_1) g_2(\lambda_2)} d M_{mn}(\lambda_1,\lambda_2)
\end{align*}
as $N \to \infty$.
Next,
\begin{align*}
&|(\mbox{the second term of (\ref{original}))} | \le \frac{1}{\left( a_{mm}^{(N,N)}(0,0)a_{nn}^{(N,N)}(0,0) \right)^{1/2}} \\
&\times \left\{ (\bm{x}_{(1,\cdot),m}^{'}, \ldots, \bm{x}_{(N,\cdot),m}^{'}) \tilde{\Sigma}^{-1} (\bm{x}_{(1,\cdot),m}^{'}, \ldots, \bm{x}_{(N,\cdot),m}^{'})^{'} \right\}^{1/2} \\
&\times \Bigg\{ \left( 
  0  ,
  \ldots  ,
 \sum_{i=1}^{N} \sum_{j=1}^{N} \pi_{i,j}^{l_1,l_2} x_{(N+1-i,N+1-j),n}
\right)  \tilde{\Sigma}^{-1} \Bigg( 
  0  ,
  \ldots  , \sum_{i=1}^{N} \sum_{j=1}^{N} \pi_{i,j}^{l_1,l_2} x_{(N+1-i,N+1-j),n}
\Bigg)^{'} \Bigg\}^{1/2} 
\end{align*}
\begin{align}
& \le \frac{1}{\left( a_{mm}^{(N,N)}(0,0)a_{nn}^{(N,N)}(0,0) \right)^{1/2}} \left( \tilde{\lambda}_{N^2} \sum_{i=1}^{N} \sum_{j=1}^{N} x_{(i,j),m}^2 \right)^{1/2} \nonumber \\
& \times \Bigg\{ \tilde{\lambda}_{N^2} \Bigg( \sum_{m_1=l_1+1-N}^{0} \sum_{m_2=1}^{l_2} \left( \sum_{i=1}^{N} \sum_{j=1}^{N} \pi_{i,j}^{m_1,m_2} x_{(N+1-i,N+1-j),n} \right)^2 \nonumber \\
& + \sum_{m_1=1}^{l_1} \sum_{m_2=l_2+1-N}^{l_2} \left( \sum_{i=1}^{N} \sum_{j=1}^{N} \pi_{i,j}^{m_1,m_2} x_{(N+1-i,N+1-j),n} \right)^2  \Bigg)
\Bigg\}^{1/2} \nonumber \\
& \le \tilde{\lambda}_{N^2} \left( \frac{N \max_{1 \le t_1,t_2 \le N} x_{(t_1,t_2),n}^2}{a_{nn}^{(N,N)}(0,0)} \right)^{1/2} \Bigg\{ \sum_{m_1=1}^{l_1} \left( \sum_{i=1}^{N} | \pi_{i,1}^{m_1} | \right)^2 \nonumber \\
& + \sum_{m_2=1}^{l_2} \left( \sum_{j=1}^{N} | \pi_{j,2}^{m_2} | \right)^2 + \sum_{m_1=1}^{l_1} \sum_{m_2=1}^{l_2} \left( \sum_{i=1}^{N} \sum_{j=1}^{N} | \pi_{i,j}^{m_1,m_2} | \right)^2 \Bigg\}, \label{second}   
\end{align}
where $\tilde{\lambda}_{N^2}$ is the greatest eigenvalue of $\tilde{\Sigma}^{-1}$ and $(\pi_{1,k}^{m_k}, \ldots, \pi_{N,k}^{m_k})^{'} \; (k = 1,2)$ is defined in the proof of  Lemma~\ref{lem1}.
From $\tilde{\Sigma} = \tilde{\Sigma}_1 \bigotimes  \tilde{\Sigma}_2$ and the property of the Kronecker product (see Horn and Johnson 1991; page 245), the eigenvalue of $\tilde{\Sigma}$ is $\lambda_{1,i}\lambda_{2,j} \; (i,j = 1,\ldots , N)$ where $\lambda_{1,i}$ and $\lambda_{2,j}$ are eigenvalues of $\tilde{\Sigma}_1$ and $\tilde{\Sigma}_2$ respectively. Because $\lambda_{1,i}$ and $\lambda_{2,j}$ are positive, 
\begin{align*}
\tilde{\lambda}_{N^2} = \Bigg\{ \left(\min_{i} \lambda_{1,i} \right) \left(\min_{j} \lambda_{2,j} \right) \Bigg\}^{-1} \le \frac{1}{ \left(\inf_{\lambda_1}g(\lambda_1) \right) \left(\inf_{\lambda_1}g(\lambda_1) \right)} < \infty,
\end{align*}
where the first inequality is derived by Proposition 4.5.3 of Brockwell and Davis (1991). By the boundedness of  $\tilde{\lambda}_{N^2}$, (\ref{second}) converges to 0 as $N \to \infty$ from (g), Lemma~\ref{lem1} and its proof.
Therefore,
\begin{align*}
 \left( D_{N^2}^{-1} X^{'} \tilde{\Sigma}^{-1} \tilde{\Sigma}(l_1,l_2) \tilde{\Sigma}^{-1} X D_{N^2}^{-1} \right)_{m,n}
 \to \frac{1}{(2\pi)^2} \int_{\Pi^2} \frac{e^{-i(l_1 \lambda_1 + l_2 \lambda_2)}}{g_1(\lambda_1) g_2(\lambda_2)} d M_{mn}(\lambda_1,\lambda_2)
\end{align*}
as $N \to \infty$. In a similar way, we can show the three cases of ($l_1 \ge 0, l_2 < 0$), ($l_1 < 0, l_2 \ge 0$) and ($l_1 < 0, l_2 < 0$). Then,
\begin{align*}
\left( D_{N^2}^{-1} X^{'} \tilde{\Sigma}^{-1} \Sigma \tilde{\Sigma}^{-1} X D_{N^2}^{-1} \right)_{m,n} 
&= \sum_{l_1=-L_1}^{L_1} \sum_{l_2=-L_2}^{L_2} b_{l_1,l_2} \left( D_{N^2}^{-1} X^{'} \tilde{\Sigma}^{-1} \tilde{\Sigma}(l_1,l_2) \tilde{\Sigma}^{-1} X D_{N^2}^{-1} \right)_{m,n} \\
&\to \frac{1}{(2\pi)^2} \int_{\Pi^2} \frac{f(\lambda_1,\lambda_2)}{g^2(\lambda_1, \lambda_2)} d M_{mn}(\lambda_1,\lambda_2)
\end{align*}
as $N \to \infty$. Finally, consider the case of general $h(\lambda_1,\lambda_2)$. From (f), (h) and Lemma~\ref{lem2}, for any sufficiently small $\epsilon > 0$, there exist $h_L(\lambda_1,\lambda_2)$ and $h_U(\lambda_1,\lambda_2)$ such that
\begin{align*}
& h_L(\lambda_1,\lambda_2) = \sum_{k_1=-K_1}^{K_1} \sum_{k_2=-K_2}^{K_2} a_{k_1,k_2} e^{i(k_1 \lambda_1 + k_2 \lambda_2)}, \quad a_{k_1,k_2} = a_{-k_1,-k_2}, \\
& h_U(\lambda_1,\lambda_2) = \sum_{k_1=-K_1}^{K_1} \sum_{k_2=-K_2}^{K_2} b_{k_1,k_2} e^{i(k_1 \lambda_1 + k_2 \lambda_2)}, \quad b_{k_1,k_2} = b_{-k_1,-k_2}, \\
& h_L(\lambda_1,\lambda_2) \le h(\lambda_1,\lambda_2) \le h_U(\lambda_1,\lambda_2) 
 \intertext{and} 
& h_U(\lambda_1,\lambda_2) - h_L(\lambda_1,\lambda_2) \le \epsilon, \quad (\lambda_1,\lambda_2) \in [-\pi,\pi]^2.
\end{align*}
Then, for any $\gamma \in \mathbb{R}^p$,
\begin{align*}
\gamma^{'} D_{N^2}^{-1} X^{'} \tilde{\Sigma}^{-1} \Sigma_L \tilde{\Sigma}^{-1} X D_{N^2}^{-1} \gamma 
\le \gamma^{'} D_{N^2}^{-1} X^{'} \tilde{\Sigma}^{-1} \Sigma \tilde{\Sigma}^{-1} X D_{N^2}^{-1} \gamma 
\le \gamma^{'} D_{N^2}^{-1} X^{'} \tilde{\Sigma}^{-1} \Sigma_U \tilde{\Sigma}^{-1} X D_{N^2}^{-1} \gamma,
\end{align*}
where $\Sigma_L$ and $\Sigma_U$ are the covariance matrices with the spectral density functions\\ $h_L(\lambda_1,\lambda_2) g(\lambda_1,\lambda_2)$ and $h_U(\lambda_1,\lambda_2) g(\lambda_1,\lambda_2)$ respectively.
From the above discussion,
\begin{align*}
&\gamma^{'}  \frac{1}{(2\pi)^2} \int_{\Pi^2} \frac{h_L(\lambda_1,\lambda_2)}{g(\lambda_1, \lambda_2)} d M(\lambda_1,\lambda_2)  \gamma 
\le \varliminf_{N \to \infty} \gamma^{'}  D_{N^2}^{-1} X^{'} \tilde{\Sigma}^{-1} \Sigma \tilde{\Sigma}^{-1} X D_{N^2}^{-1}  \gamma \\
&\le \varlimsup_{N \to \infty} \gamma^{'}  D_{N^2}^{-1} X^{'} \tilde{\Sigma}^{-1} \Sigma \tilde{\Sigma}^{-1} X D_{N^2}^{-1}  \gamma 
\le \gamma^{'}  \frac{1}{(2\pi)^2} \int_{\Pi^2} \frac{h_U(\lambda_1,\lambda_2)}{g(\lambda_1, \lambda_2)} d M(\lambda_1,\lambda_2)  \gamma.
\end{align*}
Since $\epsilon$ is arbitrary, as $N \to \infty$,
\begin{align*}
\gamma^{'}  D_{N^2}^{-1} X^{'} \tilde{\Sigma}^{-1} \Sigma \tilde{\Sigma}^{-1} X D_{N^2}^{-1}  \gamma &\to \gamma^{'}  \frac{1}{(2\pi)^2} \int_{\Pi^2} \frac{h(\lambda_1,\lambda_2)}{g(\lambda_1, \lambda_2)} d M(\lambda_1,\lambda_2)  \gamma \\
& = \gamma^{'}  \frac{1}{(2\pi)^2} \int_{\Pi^2} \frac{f(\lambda_1,\lambda_2)}{g^2(\lambda_1, \lambda_2)} d M(\lambda_1,\lambda_2)  \gamma.
\end{align*}
Because this holds for every vector $\gamma$, as $N \to \infty$,
\[
D_{N^2}^{-1} X^{'} \tilde{\Sigma}^{-1} \Sigma \tilde{\Sigma}^{-1} X D_{N^2}^{-1} \to \frac{1}{(2\pi)^2} \int_{\Pi^2} \frac{f(\lambda_1,\lambda_2)}{g^2(\lambda_1, \lambda_2)} d M(\lambda_1,\lambda_2).
\]
The proof is completed.

\hfill $\Box$

%
%
\noindent
\textbf{Acknowledgments} 
The author is grateful to Professor Ritei Shibata, Professor Yoshihiro Yajima and Professor Yasumasa Matsuda for their encouragement and helpful comments. I also acknowledge the suggestions from the editor and two anonymous referees that refined and improved the manuscript. All results remain the author's responsibility. 
%

\end{document}